\documentclass[10pt,twocolumn]{IEEEtran}
\usepackage{algorithmic}
\usepackage{graphicx}
\usepackage{textcomp}
\usepackage{xcolor}
\usepackage[nospace]{cite}
 
 \usepackage{hyperref}
   
\usepackage{url}  
\usepackage{graphicx}  
\usepackage[T1]{fontenc}
\usepackage[utf8]{inputenc}

\usepackage{caption}

\DeclareCaptionLabelFormat{lc}{{#1}~#2}
\def\tablename{Table}

\newcommand{\fr}[1]{{\mathbb{F}}\left[{#1}\right]}

 \usepackage{epsfig,endnotes}

\usepackage{floatpag,enumitem}
\setlist{leftmargin=5pt}

\usepackage{relsize}
   \usepackage{fancyhdr}
\renewcommand{\headrulewidth}{0pt}

 \usepackage{tikz}

\usepackage{multirow}
\usepackage{setspace}
\usepackage{mathrsfs}
\usepackage{bbm}
\usepackage{pifont}
\usepackage{amsfonts}

\usepackage{amsmath, amsthm}

\usepackage{tabularx}
\usepackage{listings}
\usepackage{graphicx}
\usepackage{amssymb,booktabs}
\usepackage{array}
\usepackage{cases}

 \usepackage{supertabular}

\usepackage{mathrsfs}
\usepackage{fancyhdr}

\DeclareMathOperator{\Gaussian}{Gaussian}

\usepackage{psfrag}
\usepackage{bbding}

\usepackage{float}

\usepackage{amsbsy}

\makeatletter
\providecommand{\leftsquigarrow}{%
  \mathrel{\mathpalette\reflect@squig\relax}%
}
\newcommand{\reflect@squig}[2]{%
  \reflectbox{$\m@th#1\rightsquigarrow$}%
}
\makeatother

\usepackage{algorithm}
 \usepackage{algorithmic}

\makeatletter
\newcommand{\newalgname}[1]{%
  \renewcommand{\ALG@name}{#1}%
}

\newcommand {\C} {{\rm I\kern-5.5pt C}}

\newcommand{\bp}[1]{{\mathbb{P}}\left[{#1}\right]}

\def\centerhack#1{\hbox to 0pt{\hss\footnotesize #1\hss}}
\def\centerhackn#1{\hbox to 0pt{\hss #1\hss}}
\def\dchack#1{\vbox to 0pt{\vss{\hbox to 0pt{\hss#1\hss}}\vss}}
\newcommand{\pr}[1]{{\mathbb{P}}\left[{#1}\right]} 

\usepackage{etoolbox}
\AtBeginEnvironment{align}{\setcounter{subeqn}{0}}
\newcounter{subeqn} %

\usetikzlibrary{positioning}

\newcounter{mysub}

\newtheorem{defn}{Definition}

\newtheorem{lem}{Lemma}
\newtheorem{thm}{Theorem}

\newtheorem{rem}{Remark}
\newtheorem{cor}{Corollary}

\newtheorem*{proposition1.1}{Proposition 1.1}
\newtheorem*{proposition1.2}{Proposition 1.2}
\newtheorem*{proposition1.3}{Proposition 1.3}
\newtheorem*{proposition2.1}{Proposition 2.1}
\newtheorem*{proposition2.2}{Proposition 2.2}
\usepackage{subfigure}

\begin{document}

\title{A Blockchain-Based Approach for Saving and Tracking Differential-Privacy Cost}






\author{Yang Zhao,~\IEEEmembership{Student Member,~IEEE,}
        Jun Zhao,~\IEEEmembership{Member,~IEEE,}
		Jiawen Kang,~\IEEEmembership{Member,~IEEE,}
        Zehang~Zhang,~\IEEEmembership{Student Member,~IEEE,}
        Dusit Niyato,~\IEEEmembership{Fellow,~IEEE,}
        Shuyu Shi,~\IEEEmembership{Member,~IEEE,}
        and Kwok-Yan~Lam,~\IEEEmembership{Senior Member,~IEEE}
\thanks{
Yang Zhao, Jun Zhao, Jiawen Kang, Zehang Zhang, Dusit Niyato and Kwok-Yan Lam are with School of Computer Science and Engineering, Nanyang Technological University, Singapore,  639798. (Emails: S180049@e.ntu.edu.sg, junzhao@ntu.edu.sg, kavinkang@ntu.edu.sg, prestonzzh@163.com, dniyato@ntu.edu.sg, kwokyan.lam@ntu.edu.sg).}
\thanks{Shuyu Shi is with Department of Computer Science and Engineering, Nanjing University, Nanjing, China, 210008. (Email: 
ssy@nju.edu.cn).
} 
}


\maketitle

\thispagestyle{fancy}
\pagestyle{fancy}
\lhead{This paper appears in IEEE Internet of Things Journal (IoT-J). Please feel free to contact us for questions or remarks.}
\cfoot{\thepage}
\renewcommand{\headrulewidth}{0.4pt}
\renewcommand{\footrulewidth}{0pt}

\begin{abstract}
An increasing amount of users' sensitive information is now being collected for analytics purposes. Differential privacy has been widely studied in the literature to protect the privacy of users' information. The privacy parameter bounds the information about the dataset leaked by the noisy output. Oftentimes, a dataset needs to be used for answering multiple queries, so the level of privacy protection may degrade as more queries are answered. Thus, it is crucial to keep track of privacy budget spending, which should not exceed the given limit of privacy budget. Moreover, if a query has been answered before and is asked again on the same dataset, we may reuse the previous noisy response for the current query to save the privacy cost. In view of the above, we design an algorithm to reuse previous noisy responses if the same query is asked repeatedly. In particular, considering that different requests of the same query may have different privacy requirements, our algorithm can set the optimal reuse fraction of the old noisy response and add new noise to minimize the accumulated privacy cost. Furthermore, we design and implement a blockchain-based system for tracking and saving differential-privacy cost. As a result, the owner of the dataset will have full knowledge about how the dataset has been used and be confident that no new privacy cost will be incurred for answering queries once the specified privacy budget is exhausted.
\end{abstract}


\section{Introduction} \label{sec-intro}

Massive volumes of users' sensitive information are being collected for data analytics and machine learning such as large-scale Internet of Things (IoT) data. Some IoT data contain users' confidential information, for example, energy consumption or location data. They may expose a family's habits~\cite{tudor2020bes,hassan2019privacy,xiong2018enhancing,liu2018epic,gai2019differential}. To protect personal privacy, many countries have strict policies about how technology companies collect and process users' data. However, the companies need to analyze users' data for service quality improvement. To preserve privacy while revealing useful information about datasets, 
differential privacy (DP) has been proposed~\cite{dwork2006our,dwork2006calibrating,dwork2014algorithmic}. Intuitively, by incorporating some noise, the output of an  algorithm under DP will not change significantly due to the presence or absence of one user's information in the dataset. Due to its introduction~\cite{dwork2006our,dwork2006calibrating},
DP has attracted much interest from both academia~{\cite{dwork2016concentrated,bun2016concentrated,dwork2015generalization,mcsherry2007mechanism,abadi2016deep}} and industry~\cite{tang2017privacy,erlingsson2014rappor,ding2017collecting}. For example, Apple has incorporated DP into its   mobile operating system iOS~\cite{tang2017privacy};  Google has implemented a DP tool called
RAPPOR in the Chrome browser to collect information~\cite{erlingsson2014rappor}.

Roughly speaking, a randomized mechanism achieving \mbox{$(\epsilon,\delta)$-DP}~\cite{dwork2006our} means that except with a (typically small)  probability $\delta$, altering a record in a database \mbox{cannot}  change the probability that an output is seen by more than a multiplicative factor $e^{\epsilon}$. 
Thus, the information about the dataset leaked by the noisy output of an \mbox{$(\epsilon,\delta)$-DP} algorithm is bounded by the privacy parameters $\epsilon$ and $\delta$. Smaller $\epsilon$ and $\delta$ mean stronger privacy protection and less information leakage. Note that non-zero information leakage is necessary to achieve non-zero utility. Usually, a dataset may be used for answering multiple queries (e.g., for multiple analytics tasks), thus accumulating  the information leakage and degrading  the privacy protection  level, which can be intuitively understood as the increase of privacy spending. Therefore, it is necessary to record the privacy cost to  prevent it from exceeding the privacy budget. The privacy budget is used to quantify the privacy risk when a differential private scheme is applied to real-world applications. Besides, we reduce privacy cost by reusing old noisy response to answer the current query if the query was answered before.

Traditionally, the privacy cost incurred by answering queries on a dataset is claimed by the dataset holder. Users whose information is in the dataset are not clear about the usage. It is possible that privacy consumption has exceeded the privacy budget. To solve this problem, the emerging  blockchain technology provides a new solution to manage the privacy cost. Blockchain is a chain of blocks storing  cryptographic and tamper-resistant transaction records without using a centralized server~\cite{henry2018blockchain,8832210}. With blockchain recording how the dataset is used for answering queries, users have full knowledge of how their information is analyzed. Users can easily access the blockchain to check the consumption of the privacy budget. The dataset holder has the motivation to adopt our blockchain-based approach to provide the following accountability guarantee to users whose information is in the dataset: if the dataset holder uses the dataset more than the set of queries recorded by the blockchain, measures can be taken to catch the dataset holder with cheating because transactions written into the blockchain are tamper-resistant. Yang~\emph{et~al.}~\cite{yang2017differentially} propose to leverage blockchain to track differential privacy budget, but they do not propose a mechanism to reuse noise. In contrast, we design a DP mechanism to effectively reuse previous queries' results to reuse noise and reduce privacy cost. In Section~\ref{section-yang-drawbacks}, we present more detailed comparisons.


In view of the above, we propose a blockchain-based Algorithm~\ref{algmain} and implement it to track and manage differential-privacy cost, which uses blockchain to make the privacy spending transparent to the data owner. Consequently, the data owner can track how dataset used by checking blockchain transactions' information, including each query's type, the noisy response used to answer each query, the associated noise level added to the true query result, and the remaining privacy budget. 
In addition to providing transparency of privacy management, another advantage of our blockchain-based system is as follows. Once the specified privacy budget is exhausted,  a smart contract implemented on the blockchain ensures that no new privacy cost will be incurred, and this can be verified. Furthermore, since the blockchain stores the noisy response used to answer each query, we also design an algorithm to minimize the accumulated privacy cost by reusing previous noisy response if the same query is asked again. Our algorithm (via a rigorous proof) is able to set the optimal reuse fraction of the old noisy response and add new noise (if necessary) considering different requests of the same query may be sent with different privacy requirements. 
In our blockchain-based system, 
reusing noisy responses not~only saves privacy cost, but also reduces communication overhead  when the noisy response is generated without contacting the server hosting the dataset. 




\textbf{Contributions.} The major contributions of this paper are summarized as follows:
    \begin{itemize}
      \item First, a novel privacy-preserving algorithm with a rigorous mathematical proof is designed to minimize accumulated privacy cost under a limited privacy budget by reusing previous noisy responses if the same query is received. Thus, a dataset can be used to answer more queries while preventing the privacy leakage, which is essential for the datasets with frequent queries, e.g., medical record datasets.
      \item  Second, our designed approach reduces the number of times to request the server significantly by taking advantage of recorded noisy results.
      \item Third, we implement the proposed system and algorithm according to  a detailed sequence diagram, and conduct experiments by using a real-world dataset. Numerical results demonstrate that our proposed system and algorithm are effective in saving the privacy cost while keeping accuracy.
    \end{itemize}

\textbf{Organization.}~The rest of the paper is organized as follows. Section \ref{sec-preliminaries} introduces preliminaries about differential privacy and blockchains. Section \ref{sec-system} presents system design including our proposed noise reuse algorithm. Section \ref{sec-Challenges} describes challenges in implementing our system. In Section \ref{sec-Experiments}, we discuss experimental results to validate the effectiveness of our system. Section \ref{sec-related-work} surveys related work. Section \ref{sec-conclusion} concludes this paper and identifies future directions.
 
\textbf{Notation.} Throughout the paper, $\pr{\cdot}$ denotes the probability, and $\fr{\cdot}$ stands for the probability density function. The notation $\mathcal{N}(0, A)$ denotes a Gaussian random variable with zero mean  and variance $A$, and means a fresh Gaussian noise when it is used to generate a noisy query response.  Notations used in the rest of the paper are summarized in Table~\ref{table:notation}.

\begin{table}[!h]
\caption{Summary of notations}
\centering
\begin{tabular}{|l|l|}
\hline
$(\epsilon, \delta)$              & privacy parameters                                                                                                             \\ \hline
$\pr{\cdot}$            & probability                                                                                                                \\ \hline
$\fr{\cdot}$            & probability  density  function                                                                                             \\ \hline
$D$                     & dataset                                                                                                                    \\ \hline
$D'$                    & neighbouring dataset of $D$                                                                                                      \\ \hline
$\Delta_{Q}$            & $\ell_2$-sensitivity of query $Q$                                                                                                      \\ \hline
$\sigma$                & standard deviation of the Gaussian noise                                                                                                         \\ \hline
$\widetilde{Q}_m(D)$    & \begin{tabular}[c]{@{}l@{}}noisy query response for query $Q_m$\\  on dataset $D$\end{tabular}                             \\ \hline
$Y_1, Y_2, \ldots, Y_m$ & randomized mechanisms                                                                                                      \\ \hline
$r_{\textup{optimal}}$  & the optimal fraction                                                                                                       \\ \hline
$  L_{Y}(D, D';y) $     & privacy loss                                                                                                               \\ \hline
$\mathcal{N}(0, A)$     & \begin{tabular}[c]{@{}l@{}}a Gaussian random variable with zero \\ mean  and variance $A$\end{tabular}                     \\ \hline
$V$                     & variance                                                                                                                   \\ \hline
\end{tabular}
\label{table:notation}
\end{table}

\section{Preliminaries}\label{sec-preliminaries}

We organize this section on preliminaries as follows. In Section~\ref{sec-Preliminaries-DP}, we introduce the formal definition of differential privacy. In Section~\ref{definition-blockchain}, we explain the concepts of blockchain, Ethereum and smart contract.

\subsection{Differential Privacy} \label{sec-Preliminaries-DP}

Differential privacy intuitively means that the adversary cannot determine with high confidence whether the randomized output comes from a dataset $D$ or its neighboring dataset $D'$ which differs from $D$ by one record. The formal definition of $(\epsilon,\delta)$-differential privacy is given in Definition~\ref{def-DP}, and the notion of neighboring datasets is discussed in Remark~\ref{rem-neighboring-datasets}.


 
\begin{defn}[$(\epsilon, \delta)$-Differential privacy~\cite{Dwork2014}]  \label{def-DP}
A randomized mechanism $Y$, which generates a randomized output given a dataset as the input, achieves {\textit{$(\epsilon, \delta)$-differential privacy}} if
\begin{align} 
& \pr{Y(D) \in \mathcal{Y}} \leq e^{\epsilon} \pr{Y(D')\in \mathcal{Y}} + \delta, \label{def-DP-eq}
  \\ & \textup{for $D$ and $D'$ iterating through} \nonumber \\ & \textup{all pairs of neighboring datasets, and} \nonumber  \\ & \textup{for  $\mathcal{Y}$ iterating through all subsets of the output range}, \nonumber
\end{align}
where $\pr{\cdot}$ denotes the probability, and the probability space is over the coin flips of the randomized mechanism $Y$.
\end{defn}

  \begin{rem}
  \label{rem-pureDP}
The notion of $(\epsilon,\delta)$-differential privacy under \mbox{$ \delta = 0$} becomes~{\textit{$\epsilon$-differential privacy}}. $\epsilon$-Differential privacy and $(\epsilon, \delta)$-differential privacy are also referred to as \emph{pure} and \emph{approximate} differential privacy, respectively, in many studies~\mbox{\cite{dwork2016concentrated,bun2016concentrated,dwork2015generalization}}.
 \end{rem}

\begin{rem}[\textbf{Notion of neighboring datasets}] \label{rem-neighboring-datasets}
 
Two datasets $D$ and $D'$ are called neighboring if they differ only in one tuple. There are still variants about this. In the first case, the sizes of $D$ and $D'$ differ by one so that $D'$ is obtained by adding one record to $D$ or deleting one record from $D$. In the second case, $D$ and $D'$ have the same size (say $n$), and have different records at only one of the $n$ positions. Finally, the notion of neighboring datasets can also be defined to include both the cases above. Our results in this paper
apply to all of the above cases.
 \end{rem}

 
Among various mechanisms to achieve DP, the \textit{Gaussian mechanism} for real-valued queries proposed in~\cite{dwork2006our} has received much attention. The improved result given by \cite{Dwork2014} is Lemma~\ref{lemma-Gaussian}.


 \begin{lem}[Theorem A.1 by Dwork and Roth~\cite{Dwork2014}]  \label{lemma-Gaussian}
To answer a query $Q$ with $\ell_2$-sensitivity $\Delta_{Q}$, adding a zero-mean Gaussian noise with standard deviation $\sqrt{2\ln\frac{ 1.25}{\delta}}\times\frac{\Delta_{Q}}{\epsilon}$ (denoted by $\Gaussian(\Delta_{Q}, \epsilon, \delta)$ hereafter in this paper) to each dimension of the true query result achieves $(\epsilon, \delta)$-differential privacy. The above $\ell_2$-sensitivity $\Delta_{Q}$ of a query $Q$ is defined as the maximal $\ell_2$ distance between the true query results for any two neighboring datasets $D$ and $D'$ that differ in one record; i.e., \mbox{$\Delta_{Q}  = \max_{\textrm{neighboring $D,D'$}} \|Q(D) - Q(D')\|_{2}$}. 
 \end{lem}
 
More discussions on the $\ell_2$-sensitivity of a query are given in Section~\ref{subsec-ell2-sensitivity}. Section~\ref{sec:dp-param-setting} discusses the setting of privacy parameters $\epsilon$ and $\delta$.

\subsection{Blockchain, Ethereum and Smart Contracts}~\label{definition-blockchain}
\textbf{Blockchain.} The blockchain technology is popularly used in systems requiring high security and transparency, such as Bitcoin and  Ethereum~\cite{bitcoin}. The blockchain can be effectively used to solve the double-spending problem in Bitcoin transaction by using a peer-to-peer network. The solution is to hash transaction information in a chain of hash-based Proof-of-Work (PoW, used by Bitcoin) which is the consensus mechanism algorithm used to confirm transactions and produce new blocks to the chain. Once the record is formed, it cannot be changed except redoing Proof-of-Work.

Besides, the blockchain is constantly growing with appending `completed' blocks. Blocks consisting of the most recent transactions are added to the chain in chronological order~\cite{blockchain}. Each blockchain node can have a copy of the blockchain. The blockchain allows participants to track their transactions without centralized control.



\textbf{Ethereum.} Ethereum is a blockchain platform which allows users to create decentralized end-to-end applications~\cite{wood2014Ethereum}. The miners in Ethereum use Proof-of-Work consensus algorithm  to complete transaction verification and synchronization. Besides, Ethereum can run smart contracts elaborated below.

\textbf{Smart Contract.} The smart contract was first proposed by Nick Szabo as a computerized transaction protocol that can execute terms of a contract automatically~\cite{smartcontract}. It intends to make a contract digitally, and allows to maintain credible transactions without a third party. With the development of blockchains, such as Ethereum, smart contracts are stored in the blockchain as scripts. A blockchain with a Turing-complete programming language allows everyone to customize smart contract scripts for transactions~\cite{buterin2014next}. Smart contracts are triggered when transactions are created or generated on the blockchain to finishe specific tasks or services.


\section{System Description}\label{sec-system}
Our blockchain-based system provides differentially private responses to queries while minimizing the privacy cost via noise reuse. We design a web application to implement our  Algorithm~\ref{algmain}, which generates noisy responses to queries with the minimal privacy cost by setting the optimal reuse fraction of the old noisy response and adding new noise (if necessary). For clarity, we defer Algorithm~\ref{algmain} 
and its discussion to Section~\ref{sec-system}. 
The design of the system is illustrated in Fig.~\ref{fig:system} and we discuss the details in the following. In Section~\ref{sec-Experiments}, we will discuss the implementation and experiments of our blockchain-based system, and present more figures about the implementation. In particular, Fig.~\ref{fig:screen} there shows the screenshot of our blockchain-based privacy management system~\cite{han2020blockchain}, while Fig.~\ref{fig:output} presents outputs while using the system.


\begin{figure}[!t]
\centering
\includegraphics[scale=0.4]{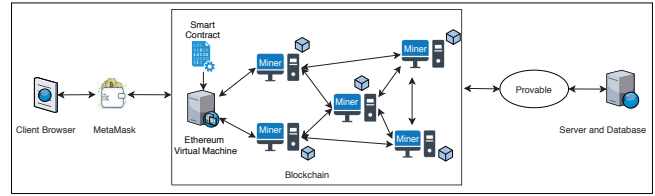}
\caption{The proposed blockchain-based system architecture for differential-privacy  cost management.}
\label{fig:system}
\end{figure}

\subsection{System Architecture}

Our system includes  the client, the blockchain, the server, and smart contract followed by more  details as  below.

\textit{Client:} The primary function of the client is to transfer users' queries to the blockchain smart contract. 
The client computes the required parameter standard deviation for the server to generate the Gaussian noise using the privacy parameters~$\epsilon$ and~$\delta$ and forwards the query to the blockchain. Also, the client can display the query result to the analyst after getting the noisy response to the query.

\textit{Blockchain Smart Contract:} The blockchain serves as a middleware between the client and the server. It decides which query should be submitted to the server. The blockchain records the remaining privacy budget, query type, the noisy response to answer the query,   the  privacy parameters, and the amount of corresponding noise. If the remaining privacy budget is enough, the smart contract will execute the query match function with the recorded   history. Otherwise, the smart contract will reject this query. If the current query does not match with any query in the history, the smart contract will call the server to calculate the result. If the query has been received before, the blockchain smart contract will not call the server if the noisy response can be completely generated by old noisy answers and will call the server if access to the dataset is still needed to generate the noisy response.

\textit{Server:} The data provider hosts the server. The server provides APIs to answer analysts' queries. When the API is called, the server will query the dataset to calculate the respective answer. After the true value $Q(D)$ is calculated, the server will add noise to perturb the answer. Then the server returns the noisy answer to the blockchain.

In the rest of the paper, we use \texttt{Blockchain}, \texttt{Client}, and \texttt{Server} to denote the blockchain, client, and server, respectively.

\subsection{System Functionality}


\textit{Match query with query history and generate noisy response:} \texttt{Blockchain} compares the current query type with saved query types to retrieve previous query results. If it is the first time for \texttt{Blockchain} to see the query, \texttt{Blockchain} will forward the query to the server, and \texttt{Server} will return the perturbed result which satisfies differential privacy to \texttt{Blockchain}. If the current query type matches previous answers' query type, \texttt{Blockchain} will compare the computed amount of noise with all previously saved amounts of noise under the same query type. Based on the comparison result, \texttt{Blockchain} will completely reuse old responses or call \texttt{Server}.

\textit{Manage privacy budget}: \texttt{Blockchain} updates the privacy budget as queries are answered and the Blockchain ensures no new privacy cost will be incurred for answering queries once the specified privacy budget is exhausted.

 \begin{algorithm*}
\setstretch{0.8}
\caption{Our proposed algorithm to answer the $m$-th query and adjust remaining privacy cost.} \label{algmain} 
 \begin{algorithmic}
 \STATE \textbf{Input:} $D$: dataset;
 ${Q}_m$: the $m$-th query;  $(\epsilon_m,\delta_m)$: requested privacy parameters for query $Q_m$; $(\sqrt{\epsilon\textup{\_squared\_remaining\_budget}},\delta_{\textup{budget}})$: remaining privacy budget (at the beginning, it is $(\sqrt{\epsilon\textup{\_squared\_budget}},\delta_{\textup{budget}})$ for $\epsilon\textup{\_squared\_budget} = {\epsilon_{\textup{budget}}}^2$);  $\Delta_{Q_m}$: $\ell_2$~sensitivity of query $Q_m$;  
\STATE \textbf{Output:} $\widetilde{Q}_m(D)$: noisy query response for query $Q_m$ on dataset $D$ under $(\epsilon_m, \delta_m)$-differential privacy;
\end{algorithmic}
\begin{algorithmic}[1]
\STATE ${\sigma}_m \leftarrow \Gaussian(\Delta_{Q_m}, \epsilon_m, \delta_m)$; \label{computesigma} //{\em \textbf{Comment:} From Lemma \ref{lemma-Gaussian}, it holds that $\Gaussian(\Delta_{Q_m}, \epsilon_m, \delta_m): \vspace{-0pt}=\sqrt{2\ln\frac{ 1.25}{\delta_m}}\times\frac{\Delta_{Q_m}}{\epsilon_m}$.}
        \IF {the query $Q_m$ is seen for the first time}
         \STATE \texttt{Client} computes $\epsilon\textup{\_squared\_cost}$ such that $ \Gaussian(\Delta_{Q_m}, \sqrt{\epsilon\textup{\_squared\_cost}}, \delta_{\textup{budget}}) = \sigma_m  $; \label{case1update1} \label{convertprivacy}  \STATE  //\vspace{-3pt}{\em \textbf{Comment:} This means \mbox{$\sqrt{2\ln\frac{ 1.25}{\delta_{\textup{budget}}}}\times\frac{\Delta_{Q_m}}{\sqrt{\epsilon\textup{\_squared\_cost}}} = \sigma_m  $, where $ \sigma_m $ as $\Gaussian(\Delta_{Q_m}, \epsilon_m, \delta_m)$ is $ \sqrt{2\ln\frac{ 1.25}{\delta_m}}\times\frac{\Delta_{Q_m}}{\epsilon_m}$.}} 
            \STATE \texttt{Client} computes $\epsilon\textup{\_squared\_remaining\_budget} \leftarrow \epsilon\textup{\_squared\_remaining\_budget} - \epsilon\textup{\_squared\_cost}$;\label{case1update2} 
         \IF {$\epsilon\textup{\_squared\_remaining\_budget}  \geq 0$}
            \STATE \textbf{return} $\widetilde{Q}_m(D) \leftarrow {Q}_m(D) + \mathcal{N} (0, 1) \times \sigma_m$; \label{alg-seen-first-time} 
            //{\em \textbf{Comment:} We refer to this Case~1) in Section~\ref{case1}. 
            If ${Q}_m$ is multidimensional, independent Gaussian noise will be added to each dimension.}
            \STATE \texttt{Blockchain} records $\langle Q_m\textup{'s query type}, \epsilon_m, \delta_m, \sigma_m, \widetilde{Q}_m(D)  \rangle$; //{\em \textbf{Comment:} This information will be kept together with a cryptographic hash of the dataset $D$, which \texttt{Blockchain} stores so it knows which records are for the same dataset $D$.}
            \ELSE
                 \STATE \textbf{return} an error of insufficient privacy budget;
            \ENDIF
            \ELSE
            \STATE Suppose $Q_m$ is a type $t$-query. \texttt{Blockchain} compares $\sigma_m$ with values in $\boldsymbol{\Sigma}_t:=\{\sigma_j: \sigma_j $  has been recorded in \texttt{Blockchain} and $  Q_j $ is a type $t$-query$\}$ (i.e., $\boldsymbol{\Sigma}_t$ consists of the corresponding noise amounts for previous instances of type $t$-query), resulting in the following subcases.
            \IF{there exists $\sigma_j \in \boldsymbol{\Sigma}_t$ such that $\sigma_m = \sigma_{j}$}
                \STATE \texttt{Blockchain} returns $\widetilde{Q}_m(D) \leftarrow \widetilde{Q}_j(D)$; \label{codeforcase2a} //{\em \textbf{Comment:} We refer to this Case~2A) in Section~\ref{case2a}}.
 \ELSIF{$\sigma_m < \min(\boldsymbol{\Sigma}_t)$}
            \STATE //{\em \textbf{Comment:} \underline{\textbf{The case of  partially reusing an old noise}}}:
          \STATE   \texttt{Client} computes $\epsilon\textup{\_squared\_cost}$ such that $[ \Gaussian(\Delta_{Q_m}, \sqrt{\epsilon\textup{\_squared\_cost}}, \delta_{\textup{budget}}) ]^{-2} = {\sigma_m}^{-2} - [\min(\boldsymbol{\Sigma}_t)]^{-2}$; \label{case2bupdate1} 
            \STATE \texttt{Client} computes $\epsilon\textup{\_squared\_remaining\_budget} \leftarrow \epsilon\textup{\_squared\_remaining\_budget} - \epsilon\textup{\_squared\_cost}$; \label{case2bupdate2} 
                    \IF {$\epsilon\textup{\_squared\_remaining\_budget} \vspace{2pt}  \geq 0$}
                \STATE
                \texttt{Blockchain} computes $\textup{NoiseReuseRatio} \leftarrow \frac{{\sigma_m}^2}{[\min(\boldsymbol{\Sigma}_t)]^2}$ and $\textup{AdditionalNoise} \leftarrow \mathcal{N}(0, 1) \times \sqrt{{\sigma_m}^2-\frac{{\sigma_m}^4}{[\min(\boldsymbol{\Sigma}_t)]^2}}$
                 \STATE \texttt{Blockchain} contacts \texttt{Server} to compute \\ $\widetilde{Q}_m(D) \leftarrow Q_m(D)+ \textup{NoiseReuseRatio} \times [\widetilde{Q}_{t,\textup{min}}(D)-Q_m(D) ] +  \textup{AdditionalNoise}$, where $\widetilde{Q}_{t,\textup{min}}(D)$ denotes the noisy response (kept in \texttt{Blockchain}) corresponding to $\min(\boldsymbol{\Sigma}_t)$; \label{codeforcase2b} //{\em \textbf{Comment:} \vspace{0pt} We refer to this Case~2B) in Section~\ref{case2b}}. 
            \STATE \texttt{Blockchain} records $\langle Q_m\textup{'s query type}, \epsilon_m, \delta_m, \sigma_m, \widetilde{Q}_m(D)  \rangle$;
                \ELSE
                 \STATE \textbf{return} an error of insufficient privacy budget;
                    \ENDIF
            \ELSE
               \STATE //{\em \textbf{Comment:} \underline{\textbf{The  case of fully reusing an old noise}}}: 
               \STATE With $\sigma_{\ell}$ denoting the maximal possible value in $\boldsymbol{\Sigma}_t$ that is also smaller than $\sigma_m$, \texttt{Blockchain} reuses $\widetilde{Q}_{\ell}(D)$, which denotes the noisy response (kept in \texttt{Blockchain}) corresponding to $\sigma_{\ell}$; 
               \STATE \texttt{Blockchain} computes $\widetilde{Q}_m(D) \leftarrow \widetilde{Q}_{\ell}(D) + \mathcal{N}(0, 1) \times \sqrt{{\sigma_m}^2 - {\sigma_{\ell}}^2}$; \label{codeforcase2c} //{\em \textbf{Comment:} We refer to this Case~2C) in Section~\ref{case2c}}.
            \STATE \texttt{Blockchain} records $\langle Q_m\textup{'s query type}, \epsilon_m, \delta_m, \sigma_m, \widetilde{Q}_m(D)  \rangle$;
        \ENDIF
         \ENDIF
\end{algorithmic}\label{algorithm}
 \end{algorithm*}

\begin{table*}[] \label{tableexample-algmain}
\scriptsize \setlength{\tabcolsep}{1.1pt} \caption{An example to explain Algorithm~\ref{algmain}.}
\label{tableexample-algmain}
\hspace{-9pt}\begin{tabular}{|l|l|l|l|l|l|l|l|l|l|l|l|l|l|}
\hline
\hspace{-1pt}$Q_m$'s query type\hspace{-1pt}                                                              & \begin{tabular}[c]{@{}l@{}}\hspace{-1pt}$Q_1\hspace{-2pt}=$\hspace{0pt}type-1\hspace{-2pt}\end{tabular} & \begin{tabular}[c]{@{}l@{}}\hspace{-1pt}$Q_2\hspace{-2pt}=$\hspace{0pt}type-2\hspace{-2pt}\end{tabular} & \begin{tabular}[c]{@{}l@{}}\hspace{-1pt}$Q_3\hspace{-2pt}=$\hspace{0pt}type-3\hspace{-2pt}\end{tabular} & \begin{tabular}[c]{@{}l@{}}\hspace{-1pt}$Q_4\hspace{-2pt}=$\hspace{0pt}type-1\hspace{-2pt}\end{tabular} & \begin{tabular}[c]{@{}l@{}}\hspace{-1pt}$Q_5\hspace{-2pt}=$\hspace{0pt}type-2\hspace{-2pt}\end{tabular} & \begin{tabular}[c]{@{}l@{}}\hspace{-1pt}$Q_6\hspace{-2pt}=$\hspace{0pt}type-1\hspace{-2pt}\end{tabular} & \begin{tabular}[c]{@{}l@{}}\hspace{-1pt}$Q_7\hspace{-2pt}=$\hspace{0pt}type-3\hspace{-2pt}\end{tabular} & \begin{tabular}[c]{@{}l@{}}\hspace{-1pt}$Q_8\hspace{-2pt}=$\hspace{0pt}type-2\hspace{-2pt}\end{tabular} & \begin{tabular}[c]{@{}l@{}}\hspace{-1pt}$Q_9\hspace{-2pt}=$\hspace{0pt}type-2\hspace{-2pt}\end{tabular} & \begin{tabular}[c]{@{}l@{}}\hspace{-1pt}$Q_{10}\hspace{-2pt}=$\hspace{0pt}type-1\hspace{-2pt}\end{tabular} & \begin{tabular}[c]{@{}l@{}}\hspace{-1pt}$Q_{11}\hspace{-2pt}=$\hspace{0pt}type-2\hspace{-2pt}\end{tabular} & \begin{tabular}[c]{@{}l@{}}\hspace{-1pt}$Q_{12}\hspace{-2pt}=$\hspace{0pt}type-1\hspace{-2pt}\end{tabular} & \begin{tabular}[c]{@{}l@{}}\hspace{-1pt}$Q_{13}\hspace{-2pt}=$\hspace{0pt}type-3\hspace{-2pt}\end{tabular} \\ \hline
\begin{tabular}[c]{@{}l@{}}\hspace{-1pt}$\sigma_m$ computed by\hspace{-1pt} \\ \hspace{-1pt}Line~\ref{computesigma} of Alg.~\ref{algmain}\hspace{-2pt}\end{tabular}                                                &
\hspace{-1pt}$\sigma_1=1$\hspace{-2pt}                                            & \hspace{-1pt}$\sigma_2=3$\hspace{-2pt}                                            & \hspace{-1pt}$\sigma_3=2$\hspace{-2pt}                                            & \hspace{-1pt}$\sigma_4=2.5$\hspace{-2pt}                                          & \hspace{-1pt}$\sigma_5=2$\hspace{-2pt}                                            & \hspace{-1pt}$\sigma_6=0.5$\hspace{-2pt}                                          & \hspace{-1pt}$\sigma_7=2$\hspace{-2pt}                                            & \hspace{-1pt}$\sigma_8=2.5$\hspace{-2pt}                                          & \hspace{-1pt}$\sigma_9=1.5$\hspace{-2pt}                                          & \hspace{-1pt}$\sigma_{10}=0.25$\hspace{-2pt}                                         & \hspace{-1pt}$\sigma_{11}=1$\hspace{-2pt}                                            & \hspace{-1pt}$\sigma_{12}=0.75$\hspace{-2pt}                                         & \hspace{-1pt}$\sigma_{13}=1.5$\hspace{-2pt}                                          \\ \hline
\begin{tabular}[c]{@{}l@{}}\hspace{-1pt}Case involved\hspace{-2pt} \\ \hspace{-1pt}in  Alg.~\ref{algmain}\hspace{-2pt}\end{tabular}                                                      & \begin{tabular}[c]{@{}l@{}}\hspace{-1pt}1):\hspace{-2pt} $\widetilde{Q}_1 \hspace{-2pt}\leftarrow\hspace{-2pt} Q_1 $\hspace{-2pt} \\ \hspace{12pt} $+$ \\ \hspace{0pt}$ \mathcal{N} (0,\hspace{-2pt} 1)\hspace{-2pt}\times \hspace{-2pt}\sigma_1  $\hspace{-1pt} \\   \hspace{-1pt}with\hspace{-2pt} \\    \hspace{-1pt}accessing \hspace{-1pt}$D$\hspace{-3pt}\end{tabular}                                                   & \begin{tabular}[c]{@{}l@{}}\hspace{-1pt}1):\hspace{-2pt} $\widetilde{Q}_2 \hspace{-2pt}\leftarrow\hspace{-2pt} Q_2 $\hspace{-2pt} \\ \hspace{12pt} $+$ \\ \hspace{0pt}$ \mathcal{N} (0,\hspace{-2pt} 1)\hspace{-2pt}\times \hspace{-2pt}\sigma_2$\hspace{-1pt}   \\   \hspace{-1pt}with\hspace{-2pt} \\    \hspace{-1pt}accessing \hspace{-1pt}$D$\hspace{-3pt}\end{tabular}            & \begin{tabular}[c]{@{}l@{}}\hspace{-1pt}1):\hspace{-2pt} $\widetilde{Q}_3 \hspace{-2pt}\leftarrow\hspace{-2pt} Q_3 $\hspace{-2pt} \\ \hspace{12pt} $+$ \\ \hspace{0pt}$ \mathcal{N} (0,\hspace{-2pt} 1)\hspace{-2pt}\times \hspace{-2pt}\sigma_3$\hspace{-1pt}   \\   \hspace{-1pt}with\hspace{-2pt} \\    \hspace{-1pt}accessing \hspace{-1pt}$D$\hspace{-3pt}\end{tabular}                                              & \begin{tabular}[c]{@{}l@{}}\hspace{-1pt}2C): $\widetilde{Q}_4 $\hspace{-2pt} \\ \hspace{-1pt}reuses $\widetilde{Q}_1$\hspace{-2pt} \\   \hspace{-1pt}without\hspace{-2pt} \\   \hspace{-1pt}accessing \hspace{-1pt}$D$\hspace{-3pt}\end{tabular}                                                   & \begin{tabular}[c]{@{}l@{}}\hspace{-1pt}2B): $\widetilde{Q}_5 $\hspace{-2pt} \\ \hspace{-1pt}reuses $\widetilde{Q}_2$\hspace{-2pt} \\   \hspace{-1pt}with\hspace{-2pt} \\    \hspace{-1pt}accessing \hspace{-1pt}$D$\hspace{-3pt}\end{tabular}                                                      & \begin{tabular}[c]{@{}l@{}}\hspace{-1pt}2B): $\widetilde{Q}_6 $\hspace{-2pt} \\ \hspace{-1pt}reuses $\widetilde{Q}_1$\hspace{-2pt} \\  \hspace{-1pt}with\hspace{-2pt} \\    \hspace{-2pt}accessing \hspace{-1pt}$D$\hspace{-3pt}\end{tabular}                                                     & \begin{tabular}[c]{@{}l@{}}\hspace{-1pt}2A): $\widetilde{Q}_7 $\hspace{-2pt} \\ \hspace{-1pt}reuses $\widetilde{Q}_3$\hspace{-2pt} \\   \hspace{-1pt}without\hspace{-2pt} \\   \hspace{-1pt}accessing \hspace{-1pt}$D$\hspace{-3pt}\end{tabular}                                                      &  \begin{tabular}[c]{@{}l@{}}\hspace{-1pt}2C): $\widetilde{Q}_8 $\hspace{-2pt} \\ \hspace{-1pt}reuses $\widetilde{Q}_5$\hspace{-2pt} \\   \hspace{-1pt}without\hspace{-2pt} \\   \hspace{-1pt}accessing \hspace{-1pt}$D$\hspace{-3pt}\end{tabular}                                                     & \begin{tabular}[c]{@{}l@{}}\hspace{-1pt}2B): $\widetilde{Q}_9 $\hspace{-2pt} \\ \hspace{-1pt}reuses $\widetilde{Q}_5$\hspace{-2pt} \\    \hspace{-1pt}with\hspace{-2pt} \\    \hspace{-2pt}accessing \hspace{-1pt}$D$\hspace{-3pt}\end{tabular}                                                      & \begin{tabular}[c]{@{}l@{}}\hspace{0pt}2B): $\widetilde{Q}_{10} $\hspace{-2pt} \\ \hspace{0pt}reuses $\widetilde{Q}_6$\hspace{-2pt} \\    \hspace{0pt}with\hspace{-2pt} \\    \hspace{0pt}accessing \hspace{-1pt}$D$\hspace{-3pt}\end{tabular}                                                       & \begin{tabular}[c]{@{}l@{}}\hspace{0pt}2B): $\widetilde{Q}_{11} $\hspace{-2pt} \\ \hspace{0pt}reuses $\widetilde{Q}_9$\hspace{-2pt} \\    \hspace{0pt}with\hspace{-2pt} \\    \hspace{0pt}accessing \hspace{-1pt}$D$\hspace{-3pt}\end{tabular}                                                        & \begin{tabular}[c]{@{}l@{}}\hspace{0pt}2C): $\widetilde{Q}_{12} $\hspace{-2pt} \\ \hspace{0pt}reuses $\widetilde{Q}_{6}$\hspace{-2pt} \\   \hspace{0pt}without\hspace{-2pt} \\   \hspace{0pt}accessing \hspace{-1pt}$D$\hspace{-3pt}\end{tabular}                                                       & \begin{tabular}[c]{@{}l@{}}\hspace{0pt}2B): $\widetilde{Q}_{13} $\hspace{-2pt} \\ \hspace{0pt}reuses $\widetilde{Q}_7$\hspace{-2pt} \\    \hspace{0pt}with\hspace{-2pt} \\    \hspace{0pt}accessing \hspace{-1pt}$D$\hspace{-3pt}\end{tabular}                                                        \\ \hline
\end{tabular}
\end{table*}


\subsection{Adversary Model} The adversary model for our system is similar to~\cite{yang2017differentially}. Assume that there are two kinds of adversaries:

First, adversaries can obtain perturbed query results.  They may try to infer users' real information using perturbed queries' results.

Second, adversaries attempt to modify the privacy budget. For example, they would like to decrease the used privacy budget so that users may exceed the privacy budget. As a result, privacy will leak. However, in our case, the privacy budget is recorded on the blockchain. The adversaries cannot tamper it once the privacy budget is stored in the blockchain.

\subsection{Our Algorithm~\ref{algmain} based on Reusing Noise} \label{subsection-explain-algmain}

We present our solution for reusing noise in Algorithm~\ref{algmain} in Section~\ref{section:algo-explain}. We consider real-valued queries so that the Gaussian mechanism can be used. Extensions to non-real-valued queries can be regarded as the future work, where we can apply the exponential mechanism of \cite{mcsherry2007mechanism}.




To clarify notation use, we note that $Q_i$ means the $i$-th query (ordered  chronologically) and is answered by a randomized algorithm $\widetilde{Q}_i$. A type $t$-query means that the query's type is $t$. Queries asked at different time can have the same query type. This is the reason that we reuse noise in Algorithm~\ref{algmain}.

Suppose a dataset $D$ has been used to answer \mbox{$m-1$}  queries $Q_1, Q_2, \ldots, Q_{m-1}$, where the $i$-th query $Q_i$ for $i=1,2,\ldots,m-1$ is answered under $(\epsilon_i, \delta_i)$-differential privacy (by reusing noise, or generating fresh noise, or combining both). For $i=1,2,\ldots,m$, we  define $\sigma_i:=\Gaussian(\Delta_{Q_i}, \epsilon_i, \delta_i)$, where $\Delta_{Q_i}$ denotes the $\ell_2$-sensitivity of $Q_i$, where we defer the discussion of $\Delta_{Q_i}$ to Section~\ref{subsec-ell2-sensitivity}.  As presented in Algorithm~\ref{algmain}, we have several cases discussed below. For better understanding of these cases, we later discuss an example given in Table~\ref{tableexample-algmain} in Section~\ref{tableexample-algmain}.

\begin{enumerate}
{\setlength\itemindent{29pt}\item[\textbf{Case~1):}]\label{case1} If $Q_{m}$ is seen for the first time, we obtain the noisy response $\widetilde{Q}_m(D)$ by adding a zero-mean Gaussian noise with standard deviation $\Gaussian(\Delta_{Q_{m}}, \epsilon_{m}, \delta_{m})$ independently to each dimension of the true result $Q_m(D)$ (if the privacy budget allows), as given by Line~\ref{alg-seen-first-time} of Algorithm~\ref{algmain}, where $\Gaussian(\Delta_{Q_m}, \epsilon_m, \delta_m): =\sqrt{2\ln\frac{ 1.25}{\delta_m}}\times\frac{\Delta_{Q_m}}{\epsilon_m}$  from Lemma~\ref{lemma-Gaussian}. 
\item[\textbf{Case~2):}]\label{case2} If $Q_{m}$ has been received before, suppose $Q_m$ is a type $t$-query, and among the previous $m-1$ queries $Q_1, Q_2, \ldots, Q_{m-1}$, let $\boldsymbol{\Sigma}_t$ consist of the corresponding noise amounts for previous instances of type $t$-query; i.e., $\boldsymbol{\Sigma}_t:=\{\sigma_j: \sigma_j $  has been recorded in \texttt{Blockchain} and $  Q_j $ is a type $t$-query$\}$. 

\texttt{Blockchain} compares $\sigma_m$ and the values in $\boldsymbol{\Sigma}_t$, resulting in the following subcases.}

\begin{enumerate}
{\setlength\itemindent{38pt}\item[\textbf{Case~2A):}]\label{case2a} If there exists $\sigma_j \in \boldsymbol{\Sigma}_t$ such that $\sigma_m = \sigma_{j}$, then $\widetilde{Q}_m(D)$ is set as $ \widetilde{Q}_j(D)$.
\item[\textbf{Case~2B):}]\label{case2b} This case considers that $\sigma_m $ is less than $\min(\boldsymbol{\Sigma}_t)$ which denotes the minimum in $\boldsymbol{\Sigma}_t$. Let $\widetilde{Q}_{t,\textup{min}}(D)$ denote the noisy response (kept in \texttt{Blockchain}) corresponding to $\min(\boldsymbol{\Sigma}_t)$; specifically, if $\min(\boldsymbol{\Sigma}_t) = \sigma_j $ for some $j$, then $\widetilde{Q}_{t,\textup{min}}(D) = \widetilde{Q}_j(D)$. Under $\sigma_m < \min(\boldsymbol{\Sigma}_t)$, to minimize the privacy cost, we reuse $\frac{{\sigma_m}^2}{[\min(\boldsymbol{\Sigma}_t)]^2}$ fraction of noise in $\widetilde{Q}_{t,\textup{min}}(D)$ to generate $\widetilde{Q}_m(D)$ (if the privacy budget allows). This will be obtained by Theorem~\ref{thm-Alg1-explain-privacy-cost-with-r}'s Result~(ii) to be presented in Section~\ref{section:algo-explain}. Specifically, under $\min(\boldsymbol{\Sigma}_t) > \sigma_m$, as given by Line~\ref{codeforcase2b} of Algorithm~\ref{algmain},  $\widetilde{Q}_m(D)$ is set by $\widetilde{Q}_m(D) \leftarrow Q_m(D)+ \frac{{\sigma_m}^2}{[\min(\boldsymbol{\Sigma}_t)]^2} \times [\widetilde{Q}_{t,\textup{min}}(D)-Q_m(D) ] +   \mathcal{N}(0, 1) \times \sqrt{{\sigma_m}^2-\frac{{\sigma_m}^4}{[\min(\boldsymbol{\Sigma}_t)]^2}}$. Note that if ${Q}_m$ is multidimensional, independent Gaussian noise will be added to each dimension according to the above formula. This also applies to other places of this paper.
\item[\textbf{Case~2C):}]\label{case2c} This case considers that $\sigma_m $ is greater than $\min(\boldsymbol{\Sigma}_t)$ and $\sigma_m $ is different from all values in $\boldsymbol{\Sigma}_t$. Let $\sigma_{\ell}$ be the maximal possible value in $\boldsymbol{\Sigma}_t$ that is also smaller than $\sigma_m$; i.e., $\sigma_{\ell} = \max \{\sigma_j: \sigma_j \in \boldsymbol{\Sigma}_t \textup{ and } \sigma_j< \sigma_m\} $. Then $\widetilde{Q}_m(D)$ is set as $\widetilde{Q}_{\ell}(D) + \mathcal{N}(0, 1) \times \sqrt{{\sigma_m}^2 - {\sigma_{\ell}}^2}$.} This will become clear by Theorem~\ref{thm-Alg1-explain-privacy-cost-with-r}'s Result~(ii) to be presented in Section~\ref{section:algo-explain}.
\end{enumerate}
\end{enumerate}

\textbf{An example to explain Algorithm~\ref{algmain}.} Table~\ref{tableexample-algmain} provides an example for better understanding of Algorithm~\ref{algmain}.  We consider three types of queries. In particular, $Q_1,Q_4,Q_6,Q_{10},Q_{12}$ are type $1$-queries; $Q_2,Q_5,Q_8,Q_9,Q_{11}$ are  type $2$-queries, and $Q_3,Q_7,Q_{13}$ are  type $3$-queries.

 \subsection{Explaining the Noise Reuse Rules of Algorithm~\ref{algmain}} \label{section:algo-explain}

Our noise-reuse rules of Algorithm~\ref{algmain} are designed to minimize the accumulated privacy cost. 
To explain this, inspired by~\cite{abadi2016deep}, we define the privacy loss to quantify privacy cost.  We analyze the privacy loss to characterize how privacy degrades in a fine-grained manner, instead of using the composition theorem by Kairouz~\emph{et~al.}~\cite{kairouz2017composition}. Although \cite{kairouz2017composition} gives the state-of-the-art results for the composition of differentially private algorithms, the results do not assume the underlying mechanisms to achieve differential privacy. In our analysis, by analyzing  the privacy loss of Gaussian mechanisms specifically, we can obtain smaller privacy cost.

For a randomized algorithm $Y$, neighboring datasets $D$ and $D'$, and output $y$, the privacy loss $  L_{Y}(D, D';y) $ represents the multiplicative difference between the probabilities that the same output $y$ is observed when the randomized algorithm $Y$ is applied to $D$ and $D'$. Specifically, we define 
\begin{align}
  L_{Y}(D, D';y)  := \ln \frac{\fr{Y(D)=y}}{\fr{Y(D')=y}}, \label{eqn-L-Y-D-Dprime}
\end{align}
where $\fr{\cdot}$ denotes the probability density function.

For simplicity, we use probability density function $\fr{\cdot}$ in Eq.~(\ref{eqn-L-Y-D-Dprime}) above by assuming that the randomized algorithm $Y$ has the continuous output. If $Y$ has the discrete output, we replace $\fr{\cdot}$ by probability mass function $\bp{\cdot}$.

When $y$ follows the probability distribution of random variable $Y(D)$, $  L_{Y}(D, D';y) $ follows the probability distribution of random variable ${L}_{Y}(D, D';Y(D))$, which we write as ${L}_{Y}(D, D')$ for simplicity. 






We denote the composition of some randomized mechanisms $Y_1, Y_2, \ldots, Y_m$ for a positive integer $m$ by $Y_1 \Vert Y_2 \Vert  \ldots \Vert Y_m$.  
For the composition, the privacy loss with respect to neighboring datasets $D$ and $D'$ when the outputs of randomized mechanisms $Y_1, Y_2, \ldots, Y_m$ are $y_1, y_2, \ldots, y_m$ is defined by
\begin{align} 
&  L_{Y_1 \Vert Y_2 \Vert  \ldots \Vert Y_m}(D, D';y_1, y_2, \ldots, y_m)   \nonumber
  \\ &:= 
  \ln \frac{ \mathbb{F}\big[\cap_{i=1}^m \left[Y_i(D) = y_i\right]\big] }{  \mathbb{F}\big[\cap_{i=1}^m \left[Y_i(D') = y_i \right]\big]}. \nonumber
 \end{align}
When $y_i$ follows the probability distribution of random variable $Y_i(D)$ for each $i\in \{1,2,\ldots, m\}$, clearly $L_{Y_1 \Vert Y_2 \Vert  \ldots \Vert Y_m}(D, D';y_1, y_2, \ldots, y_m)  $ follows the probability distribution of random variable $L_{Y_1 \Vert Y_2 \Vert  \ldots \Vert Y_m}(D, D';Y_1(D), Y_2(D), \ldots, Y_m(D)) $, which we write as $L_{Y_1 \Vert Y_2 \Vert  \ldots \Vert Y_m}(D, D')$ for simplicity.

With the privacy loss defined above, we now analyze how to reuse noise when a series of queries are answered under differential privacy. To this end, we present Theorem~\ref{thm-Alg1-explain-privacy-cost-with-r}, which presents the optimal ratio of reusing noise to minimize privacy cost.


\begin{thm}[\textbf{Optimal ratio of reusing noise to minimize privacy cost}]  \label{thm-Alg1-explain-privacy-cost-with-r}
Suppose that before answering query $Q_m$ and after answering $Q_1, Q_2, \ldots, Q_{m-1}$, the privacy loss $L_{\widetilde{Q}_1 \Vert  \widetilde{Q}_2 \Vert   \ldots  \Vert   \widetilde{Q}_{m-1}}(D, D')$ is given by $\mathcal{N}(\frac{A(D, D')}{2},A(D, D'))$ for some $A(D, D')$. For the $m$-th query $Q_m$, suppose that $Q_m$ is the same as $Q_j$ for some $j\in\{1,2,\ldots,m-1\}$ and we reuse $r$  fraction of noise in $\widetilde{Q}_{j}(D)$ to generate $\widetilde{Q}_m(D)$ for $0\leq r \leq 1$ satisfying ${\sigma_m}^2 - r^2 {\sigma_j}^2>0$, where $r$ is a constant to be decided. If $\widetilde{Q}_{j}(D)-Q_j(D)$ follows a Gaussian probability distribution with mean $0$ and standard deviation $\sigma_j$, we generate the noisy response $\widetilde{Q}_m(D)$ to answer query $Q_m$ as follows:
\begin{align} 
&\widetilde{Q}_m(D)  \nonumber
\\ &  \leftarrow Q_m(D)+ r [\widetilde{Q}_j(D) -Q_j(D) ] + \mathcal{N}(0,{\sigma_m}^2 - r^2 {\sigma_j}^2), \label{reuse-expr2}  
\end{align} 
so that $\widetilde{Q}_{m}(D)-Q_m(D)$ follows a Gaussian probability distribution with mean $0$ and standard deviation $\sigma_m$.

Note that $\Delta_{Q_m}$ and $\Delta_{Q_j}$ are the same since $Q_m$ and $Q_j$ are the same. Then we have the following results.
\begin{itemize}
\item[(i)]
After answering the $m$ queries $Q_1, Q_2, \ldots, Q_{m}$, the privacy loss $L_{\widetilde{Q}_1 \Vert  \widetilde{Q}_2 \Vert   \ldots  \Vert   \widetilde{Q}_m}(D, D')$ will be $\mathcal{N}(\frac{B_r(D, D')}{2},B_r(D, D'))$ for  $B_r(D, D'): = A(D, D') + \frac{[\| Q_m(D) -Q_m(D') \|_2]^2 (1 - r )^2}{{\sigma_m}^2 - r^2 {\sigma_j}^2}$.
\item[(ii)] We clearly require $ r \geq 0$ and ${\sigma_m}^2 - r^2 {\sigma_j}^2 \geq 0$ in~(\ref{reuse-expr2}) above (note that $\mathcal{N}(0,0)\equiv 0$). To minimize the total privacy cost (which is equivalent to minimize $B_r(D, D')$ above), the optimal $r$ is given by 
\begin{align} 
 r_{\textup{optimal}} = \begin{cases}  1, &\textup{if } \sigma_m \geq \sigma_j , \\[2pt]
 \big( \frac{\sigma_m}{\sigma_j} \big)^2, &\textup{if } \sigma_m < \sigma_j,
\end{cases} 
\label{roptimal}  
\end{align} 
so that substituting Eq.~(\ref{roptimal})  into the expression of $B_r(D, D')$ gives
\begin{align} 
 & B_{r_{\textup{optimal}}}(D, D')  \nonumber
\\ & \hspace{-2pt}  =\hspace{-1pt}  \begin{cases} \hspace{-2pt} A(D, D'),  \textup{ if } \sigma_m \geq  \sigma_j ; \\[2pt]
\hspace{-2pt} A(D, D') +  [\| Q_m(D) -Q_m(D') \|_2]^2 \left(\hspace{-1pt}  \frac{1}{{\sigma_m}^2} -  \frac{1}{{\sigma_j}^2} \hspace{-1pt} \right)\hspace{-2pt}, \\ ~~~~~~~~~~~\hspace{7pt}\textup{if } \sigma_m < \sigma_j.
\end{cases} 
\label{Brresult}  
\end{align}  
Note that if $\sigma_m = \sigma_j$ for some $j\in\{1,2,\ldots,m-1\}$, we have $r_{\textup{optimal}}=1$ and just set $\widetilde{Q}_m(D)$ as $\widetilde{Q}_j(D)$.
\end{itemize}
\end{thm}
\begin{proof}
The proof is in Appendix~\ref{sec-proof-theorem-1}.
\end{proof}


Eq.~(\ref{roptimal}) of
Theorem~\ref{thm-Alg1-explain-privacy-cost-with-r} clearly indicates the noise use ratio $\frac{{\sigma_m}^2}{[\min(\boldsymbol{\Sigma}_t)]^2}$ of Case~2B) in Algorithm~\ref{algmain} (see Line~\ref{codeforcase2b} of Algorithm~\ref{algmain}), and the noise use ratio $1$ of Cases~2A) and~2C) in Algorithm~\ref{algmain} (see Lines~\ref{codeforcase2a} and~\ref{codeforcase2c} of Algorithm~\ref{algmain}).


By considering $r=0$ in Result (i) of Theorem~\ref{thm-Alg1-explain-privacy-cost-with-r}, we obtain Corollary~\ref{cor-Alg1-explain-privacy-cost-with1}, which presents the classical result on the privacy loss of a single run of the Gaussian mechanism.

\begin{cor}[]  \label{cor-Alg1-explain-privacy-cost-with1}

By considering $m=1$ in Result (i) of Theorem~\ref{thm-Alg1-explain-privacy-cost-with-r}, we have that for a randomized algorithm $\widetilde{Q}$ which adds Gaussian noise amount $\sigma$ to a query $Q$, the privacy loss with respect to neighboring datasets $D$ and $D'$ is given by $\mathcal{N}(\frac{A(D, D')}{2},A(D, D'))$ for $A(D, D') : = \frac{[\| Q(D) -Q(D') \|_2]^2}{{\sigma}^2} $.
\end{cor}

Corollary~\ref{cor-Alg1-explain-privacy-cost-with1} has been shown in many prior studies~\cite{Dwork2014,dwork2016concentrated,bun2016concentrated} on the Gaussian mechanism for differential privacy.

By considering $r=0$ in Result (i) of Theorem~\ref{thm-Alg1-explain-privacy-cost-with-r}, we obtain Corollary~\ref{cor-Alg1-explain-privacy-cost-with2}, which presents the privacy loss of the naive algorithm where the noisy response to each query is generated independently using fresh noise.

\begin{cor}[\textbf{Privacy loss of the naive algorithm where each query is answered independently}] \label{cor-Alg1-explain-privacy-cost-with2}
 Suppose a dataset has been used to answer $n$ queries $Q_1, Q_2, \ldots, Q_n$ under differential privacy. Specifically, for $i=1,2,\ldots,n$, to answer the $i$-th query $Q_i$ under $(\epsilon_i, \delta_i)$-differential privacy, a noisy response $\widetilde{Q}_i $ is generated by adding independent Gaussian noise $\sigma_i:=\Gaussian(\Delta_{Q_i}, \epsilon_i, \delta_i)$ to the true query result ${Q}_i $, where $\Delta_{Q_i}$ is the $\ell_2$-sensitivity of $Q_i$. 
Then after answering $n$ queries $Q_1, Q_2, \ldots, Q_n$ independently as above, the privacy loss with respect to neighboring datasets $D$ and $D'$  is given by $\mathcal{N}(\frac{F(D, D')}{2},F(D, D'))$ for $ F(D, D'):= \sum_{i=1}^n \frac{[\| Q_i(D) -Q_i(D') \|_2]^2}{{\sigma_i}^2}$.
\end{cor}


\subsection{Explaining Privacy Cost Update in Algorithm~\ref{algmain}}

Among the above cases, Cases~2A) and 2C) do not incur additional privacy cost since they just use previous noisy results and generate fresh Gaussian noise, without access to the dataset $D$. In contrast,  Cases~1) and 2B) incur additional privacy cost since they need to access the dataset $D$ to compute the true query result $Q_m(D)$. Hence, in Algorithm~\ref{algmain}, the privacy cost is updated in Cases~1) and 2B), but not in Cases~2A) and 2C). In this section, we explain the reason that the privacy cost is updated in Algorithm~\ref{algmain} according to Lines~\ref{case1update1} and~\ref{case1update2} for Case~1), and Lines~\ref{case2bupdate1} and~\ref{case2bupdate2} for Case~2B).


When our Algorithm~\ref{algmain} is used, we let the above randomized mechanism $Y_i$ be our noisy response function
$\widetilde{Q}_i$. When $\widetilde{Q}_1, \widetilde{Q}_2, \ldots, \widetilde{Q}_{i-1}$ on dataset $D$ are instantiated as $y_1, y_2, \ldots, y_{i-1}$, if the generation of $\widetilde{Q}_i$ on dataset $D$ uses $\widetilde{Q}_j$ for some $j<i$, then the auxiliary information $\textup{aux}_i$ in the input to $\widetilde{Q}_i$ contains $y_j$ ($\textup{aux}_1$ is $\emptyset$). For the consecutive use of our Algorithm~\ref{algmain}, it will become clear that the privacy loss, defined by $L_{\widetilde{Q}_1, \widetilde{Q}_2, \ldots, \widetilde{Q}_m}(y_1, y_2, \ldots, y_m)  := \ln \max_{\textup{neighboring datasets $D,D'$}} \frac{ \fr{\cap_{i=1}^m \left[\widetilde{Q}_i(D) = y_i\right]} }{ \fr{\cap_{i=1}^m \left[\widetilde{Q}_i(D') = y_i \right]}}$, follows a Gaussian probability distribution with mean $\frac{V}{2}$ and variance $V$ for some $V$, denoted by $\mathcal{N}(\frac{V}{2},V)$. For such a reason that  form of privacy loss, the corresponding differential-privacy level is given by the following lemma.

\begin{lem} \label{lem-privacy-cost-DP-condition}
If the privacy loss of a randomized mechanism $Y$ with respect to neighboring datasets $D$ and $D'$ is given by $\mathcal{N}(\frac{V(D, D')}{2},V(D, D'))$ for some $V(D, D')$, then $Y$ achieves $(\epsilon, \delta)$-differential privacy for $\epsilon$ and $\delta$ satisfying   $\max_{\textup{neighboring datasets $D,D'$}} V(D, D') =  [\Gaussian(1, \epsilon, \delta)]^{-2} $.
\end{lem}

\begin{proof}
The proof details are in Appendix~\ref{sec-proof-lemma-2}.
\end{proof}


Based on the privacy loss defined above, we have the following theorem which explains the rules to update the privacy cost in our Algorithm~\ref{algmain}.

\begin{thm} \label{Alg1-explain-privacy-cost}
We consider the consecutive use of Algorithm~\ref{algmain} here. Suppose that after answering $Q_1, Q_2, \ldots, Q_{m-1}$ and before answering query $Q_m$, the privacy loss with respect to neighboring datasets $D$ and $D'$ is given by $\mathcal{N}(\frac{A(D, D')}{2}, A(D, D'))$ for some $A(D, D')$, and the corresponding privacy level can be given by $(\epsilon_{\textup{old}}, \delta_{\textup{budget}})$-differential privacy. Then in Algorithm~\ref{algmain}, after answering all $m$ queries $Q_1, Q_2, \ldots, Q_{m-1}, Q_m$, we have: 
\begin{itemize}
\item
the privacy loss with respect to neighboring datasets $D$ and $D'$
\begin{itemize}
\item[\ding{172}] will still be $\mathcal{N}(\frac{A(D, D')}{2},A(D, D'))$ in Cases~2A) and 2C), 
\item[\ding{173}] will be  $\mathcal{N}(\frac{B(D, D')}{2},B(D, D'))$ in Case~1) for $B(D, D'):=A(D, D') + \frac{[\| Q_m(D) -Q_m(D') \|_2]^2}{{\sigma_m}^2} $,
\item[\ding{174}] will be $\mathcal{N}(\frac{C(D, D')}{2},C(D, D'))$ in Case~2B) \vspace{1pt} for {$C(D, D'):=A(D, D') + [\| Q_m(D) -Q_m(D') \|_2]^2 \times \left[ \frac{1}{{\sigma_m}^2} - \frac{1}{[\min(\boldsymbol{\Sigma}_t)]^2}\right]$;} 
\end{itemize}
\item
the corresponding privacy level can be given by $(\epsilon_{\textup{new}}, \delta_{\textup{budget}})$-differential privacy with the following $\epsilon_{\textup{new}}$:
\begin{itemize}
\item[\ding{175}] $\epsilon_{\textup{new}} = \epsilon_{\textup{old}}$ in Cases~2A) and 2C), 
\item[\ding{176}] ${\epsilon_{\textup{new}}}^2 = {\epsilon_{\textup{old}}}^2 + \epsilon\textup{\_squared\_cost}$ in Case~1) for $\epsilon\textup{\_squared\_cost}$ satisfying $ \Gaussian(\Delta_{Q_m}, \sqrt{\epsilon\textup{\_squared\_cost}}, \delta_{\textup{budget}}) = \sigma_m  $,
\item[\ding{177}] ${\epsilon_{\textup{new}}}^2 = {\epsilon_{\textup{old}}}^2 + \epsilon\textup{\_squared\_cost}$ in Case~2B) for $\epsilon\textup{\_squared\_cost}$ satisfying $[ \Gaussian(\Delta_{Q_m}, \sqrt{\epsilon\textup{\_squared\_cost}}, \delta_{\textup{budget}}) ]^{-2} = {\sigma_m}^{-2} - [\min(\boldsymbol{\Sigma}_t)]^{-2}$. 
\end{itemize}
\end{itemize}
\end{thm}

Theorem~\ref{Alg1-explain-privacy-cost} explains the rules to update the privacy cost in Algorithm~\ref{algmain}. Specifically, Result~\ding{176} gives Lines~\ref{case1update1} and~\ref{case1update2} for Case~1), and Result~\ding{177} gives Lines~\ref{case2bupdate1} and~\ref{case2bupdate2} for Case~2B).

\begin{proof}
The proof is in Appendix~\ref{sec-proof-theorem-2}.
\end{proof}

\subsection{Analyzing the Total Privacy Cost}

Based on Theorem~\ref{Alg1-explain-privacy-cost}, we now analyze the total privacy cost when our system calls Algorithm~\ref{algmain} consecutively.

At the beginning when no query has been answered, we have $V=0$ (note that $\mathcal{N}(0,0)\equiv 0$). Then by induction via Corollary~\ref{cor-Alg1-explain-privacy-cost-with1} and Theorem~\ref{Alg1-explain-privacy-cost}, for the consecutive use of Algorithm~\ref{algmain},
the privacy loss is always in the form of $\mathcal{N}(\frac{V}{2},V)$ for some $V$.
 In our Algorithm~\ref{algmain}, the privacy loss changes only when the query being answered belongs to Cases~1) and~2B). More formally, we have the following theorem.

\begin{thm} \label{Alg1-total-privacy-cost}
Among queries $Q_1, Q_2, \cdots, Q_n$, let $N_1$,  $N_{2A}$, $N_{2B}$, and $N_{2C}$ be the set of $i \in \{1,2,\ldots,n\}$ such that $Q_i$ is in Cases~1),~2A),~2B), and~2C), respectively. For queries in Case~2B), let $T_{2B}$ be the set of query types. In Case~2B), for query type $t\in T_{2B}$, suppose the number of type-$t$ queries be $m_t$, and let these type-$t$ queries be $Q_{j_{t,1}}, Q_{j_{t,2}}, \ldots, Q_{j_{t,m_t}}$ for indices $j_{t,1},j_{t,2},\ldots ,j_{t,m_t} $ (ordered  chronologically)  all belonging to $N_{2B}$. From Case~2B) of Algorithm~\ref{algmain}, we have $\sigma_{j_{t,1}} > \sigma_{j_{t,2}} > \ldots > \sigma_{j_{t,m_t}}$,  and for $k\in \{2, 3,\ldots, m_t\}$, $\widetilde{Q}_{j_{t,k}}$ is answered by reusing $\frac{{\sigma_{j_{t,k}}}^2}{{\sigma_{j_{t,k-1}}}^2}$ fraction of old noise in $\widetilde{Q}_{j_{t,k-1}}$; more specifically, $\widetilde{Q}_{j_{t,k}} = {Q}_{j_{t,k}} + \frac{{\sigma_{j_{t,k}}}^2}{{\sigma_{j_{t,k-1}}}^2}[\widetilde{Q}_{j_{t,k-1}} - {Q}_{j_{t,k-1}}] + \mathcal{N}(0, {\sigma_{j_{t,k}}}^2-\frac{{\sigma_{j_{t,k}}}^4}{{\sigma_{j_{t,k-1}}}^2})$ from Line~\ref{codeforcase2b} of Algorithm~\ref{algmain} in Section~\ref{section:algo-explain} for Case~2B). We also consider that $\widetilde{Q}_{j_{t,1}}$ is answered by reusing $\frac{{\sigma_{j_{t,1}}}^2}{{\sigma_{j_{t,0}}}^2}$ fraction of old noise in $\widetilde{Q}_{j_{t,0}}$. Let the $\ell_2$-sensitivity of a type-$t$ query be $\Delta(\textup{type-}t)$.

In the example provided in Table~\ref{tableexample-algmain}, we have $N_1 = \{1,2,3\} $,  $N_{2A}= \{7\} $, $N_{2B}= \{5,6,9,10,11,13\} $, and $N_{2C}= \{8,12\} $. $T_{2B} = \{\textup{type-}1,\textup{type-}2,\textup{type-}3\}$. In Case~2B),  the number of type-$1$ queries is $m_1 = 2$, and these type-$1$ queries are $Q_6$ and $Q_{10}$ so $j_{1,1} = 6$ and $j_{1,2} = 10$ (also $j_{1,0} = 1$ since $\widetilde{Q}_6 $ reuses $\widetilde{Q}_1$); the number of type-$2$ queries is $m_2 = 3$, and these type-$2$ queries are $Q_5, Q_9$, and $Q_{11}$ so $j_{2,1} = 5$ and $j_{2,2} = 9$, $j_{2,3} = 11$ (also $j_{2,0} = 2$ since $\widetilde{Q}_5 $ reuses $\widetilde{Q}_2$); the number of type-$3$ queries is $m_3 = 1$, and this type-$3$ query is $Q_{13}$ so $j_{3,1} = 13$ (also $j_{3,0} = 3$ since $\widetilde{Q}_{13} $ reuses $\widetilde{Q}_3$).

Then after Algorithm~\ref{algmain} is used to answer all $n$ queries with query $Q_i$ being answered under $(\epsilon_i, \delta_i)$-differential privacy,  we have:
\begin{itemize}
\item The total privacy loss with respect to neighboring datasets $D$ and $D'$ is given by $\mathcal{N}(\frac{G(D, D')}{2},G(D, D'))$, where
\begin{align} 
& G(D, D'):= \sum_{i \in N_1} \frac{[\| Q_i(D) -Q_i(D') \|_2]^2}{{\sigma_i}^2} \nonumber
\\ & \quad + \sum_{t\in T_{2B}} \Bigg\{\frac{{[\| Q_{j_{t,m_t}}(D) -Q_{j_{t,m_t}}(D') \|_2]}^2}{{\sigma_{j_{t,m_t}}}^2} \nonumber
\\ & \quad \quad \quad  \quad \quad - \frac{{[\| Q_{j_{t,0}}(D) -Q_{j_{t,0}}(D') \|_2]}^2}{{\sigma_{j_{t,0}}}^2}\Bigg\} , \label{label-eq-G}  
\end{align} 
 and the first summation is the contribution from queries in Case~1), and the second summation is the contribution from queries in Case~2B). When $D$ and $D'$ iterate the space of neighboring datasets, the maximum of $\| Q_i(D) -Q_i(D') \|$ is $Q_i$'s $\ell_2$-sensitivity $\Delta_{Q_i}$, and the maximum of both $\| Q_{j_{t,m_t}}(D) -Q_{j_{t,m_t}}(D') \|_2$ and $\| Q_{j_{t,0}}(D) -Q_{j_{t,0}}(D') \|_2$ are $\Delta(\textup{type-}t)$ since $Q_{j_{t,m_t}}$ and $Q_{j_{t,0}}$ are both type-$t$ queries, we obtain
 \begin{align} 
 & \max_{\textup{neighboring datasets $D,D'$}} G(D, D') \nonumber
\\ & = \sum_{i \in N_1} \frac{{\Delta_{Q_i}}^2}{{\sigma_i}^2} + \sum_{t\in T_{2B}} \left[\frac{{[\Delta(\textup{type-}t)]}^2}{{\sigma_{j_{t,m_t}}}^2} - \frac{{[\Delta(\textup{type-}t)]}^2}{{\sigma_{j_{t,0}}}^2}\right]. \label{label-eq-G-maxDDprime}  
\end{align} 
 
In the example provided in Table~\ref{tableexample-algmain} in Section~\ref{tableexample-algmain}, $\max_{\textup{neighboring datasets $D,D'$}} G(D, D')$ is given by $\frac{{\Delta_{Q_1}}^2}{{\sigma_1}^2} + \frac{{\Delta_{Q_2}}^2}{{\sigma_2}^2} + \frac{{\Delta_{Q_3}}^2}{{\sigma_3}^2} + \left[\frac{{[\Delta(\textup{type-}1)]}^2}{{\sigma_{10}}^2} - \frac{{[\Delta(\textup{type-}1)]}^2}{{\sigma_{1}}^2}\right]  + \left[\frac{{[\Delta(\textup{type-}2)]}^2}{{\sigma_{11}}^2} - \frac{{[\Delta(\textup{type-}2)]}^2}{{\sigma_{2}}^2}\right]  + \left[\frac{{[\Delta(\textup{type-}3)]}^2}{{\sigma_{13}}^2} - \frac{{[\Delta(\textup{type-}3)]}^2}{{\sigma_{3}}^2}\right] = \frac{{[\Delta(\textup{type-}1)]}^2}{{\sigma_{10}}^2}  + \frac{{[\Delta(\textup{type-}2)]}^2}{{\sigma_{11}}^2}  + \frac{{[\Delta(\textup{type-}3)]}^2}{{\sigma_{13}}^2}.  $
\item From Lemma~\ref{lem-privacy-cost-DP-condition}, the total privacy cost of our Algorithm~\ref{algmain} can be given by $(\epsilon_{\textup{ours}}, \delta_{\textup{budget}})$-differential privacy for $\epsilon_{\textup{ours}}$  satisfying
\begin{align} 
\hspace{-12pt}[\Gaussian(1, \epsilon_{\textup{ours}}, \delta_{\textup{budget}})]^{-2}= \max_{\textup{neighboring datasets $D,D'$}} G(D, D') , \label{label-eq-G-maxDDprime-privcost}  
\end{align} 
or  $(\epsilon, \delta)$-differential privacy for any $\epsilon$ and $\delta$ satisfying  $[\Gaussian(1, \epsilon, \delta)]^{-2}= \max_{\textup{neighboring datasets $D,D'$}} G(D, D') $. 
\end{itemize}
 
\end{thm}

\begin{proof}
The proof is in Appendix~\ref{sec-proof-theorem-3}.
\end{proof}


\begin{rem} Theorem~\ref{Alg1-total-privacy-cost} can be used to understand that our Algorithm~\ref{algmain} incurs less privacy cost than that of the naive algorithm where $n$ queries are answered independently.
 As given in Corollary~\ref{cor-Alg1-explain-privacy-cost-with2}, the privacy loss with respect to neighboring datasets $D$ and $D'$ is given by $\mathcal{N}(\frac{F(D, D')}{2},F(D, D'))$ for $ F(D, D'):= \sum_{i=1}^n \frac{[\| Q_i(D) -Q_i(D') \|_2]^2}{{\sigma_i}^2}$. Clearly, $F(D, D')\geq G(D, D')$ for $G(D, D')$ given by Eq.~(\ref{label-eq-G}) above. From Lemma~\ref{lem-privacy-cost-DP-condition}, the privacy cost of the naive algorithm can be given by $(\epsilon_{\textup{naive}}, \delta_{\textup{budget}})$-differential privacy for $\epsilon_{\textup{naive}}$  satisfying $[\Gaussian(1, \epsilon_{\textup{naive}}, \delta_{\textup{budget}})]^{-2}= \max_{\textup{neighboring datasets $D,D'$}} F(D, D')$, which with Eq.~(\ref{label-eq-G-maxDDprime-privcost}) in Theorem~\ref{Alg1-total-privacy-cost} and the expression of $\Gaussian(\cdot, \cdot, \cdot)$ in Lemma~\ref{lemma-Gaussian} implies $\frac{{\epsilon_{\textup{ours}}}}{{\epsilon_{\textup{naive}}}}    = \sqrt{\frac{\max_{\textup{neighboring datasets $D,D'$}} G(D, D')}{\max_{\textup{neighboring datasets $D,D'$}} F(D, D')} }  \leq 1$ , where the equal sign is taken only when all $n$ queries are different so no noise reuse is incurred in our Algorithm~\ref{algmain}.
 \end{rem}

 

\subsection{Computing the $\ell_2$-sensitivity of A Query} \label{subsec-ell2-sensitivity}

The $\ell_2$-sensitivity of a query $Q$ is defined as the maximal $\ell_2$ distance between the (true) query results for any  neighboring datasets $D$ and $D'$ that differ in one record: \mbox{$\Delta_{Q}  = \max_{\textup{neighboring datasets $D,D'$}} \|Q(D) - Q(D')\|_{2}$}. For one-dimensional real-valued query $Q$, $\Delta_{Q}$ is simply the maximal absolute difference between $Q(D)$ and $Q(D')$ for any neighboring datasets $D$ and $D'$. In Section~\ref{sec-Experiments} for performance evaluation,   we define neighboring datasets by considering modifying an entry. Then if the dataset has $n$ users' information, and the domain of each user's income is within the interval $[\textup{min\_income}, \textup{max\_income}]$, then $\Delta_{Q}$ for query $Q$ being the average income of all users is $\frac{\textup{max\_income}-\textup{min\_income}}{n}$ since this is the maximal variation in the output when a user's record changes. Similarly, $\Delta_{Q}$ for query $Q$ being the percentage of female users is  $\frac{1}{n}$.

 

\section{Implementation Challenges of Our Blockchain-Based System} \label{sec-Challenges}

We now discuss  challenges and countermeasures during the design and implementation of our blockchain-based system.
 

\textit{Smart Contract fetches external data.}
Ethereum blockchain applications, such as Bitcoin scripts and smart contracts are unable to access and fetch directly the external data they need. However, in our application, \texttt{Blockchain} needs to fetch data from \texttt{Server} then returns them to \texttt{Client}. This requires smart contract to send the HTTP POST request. Hence, we use the Provable, a service integrated with a number of blockchain protocols and can be accessed by non-blockchain applications as well. It guarantees that data fetched from the original data-source is genuine and untampered.

By using the Provable, smart contracts can directly access data from web sites or APIs. In our case, \texttt{Blockchain} can send HTTP requests to \texttt{Server} with parameters, and then process and store data after \texttt{Server} responds successfully.

\textit{Mathematical operations with Solidity.}
\texttt{Blockchain} is written using solidity language which is designed to target Ethereum Virtual Machine. However, current solidity language does not have inherent functions for complex mathematical operations, such as taking the square root  or  logarithm. We  write a function to implement the square root operation. To avoid using Lemma~\ref{lemma-Gaussian} to compute logarithm in \texttt{Blockchain}, we generate Gaussian noise in \texttt{Client}, and pass the value to \texttt{Blockchain} as one of the parameters in function QueryMatch.  Besides, current Solidity version cannot operate float or double type data. To keep the precision, we scale up the noise amount during calculation, and then scale down the value before returning the noisy data to analysts. 


\begin{figure}[!t]
\centering
\includegraphics[width=0.5\textwidth,height=5.5cm]{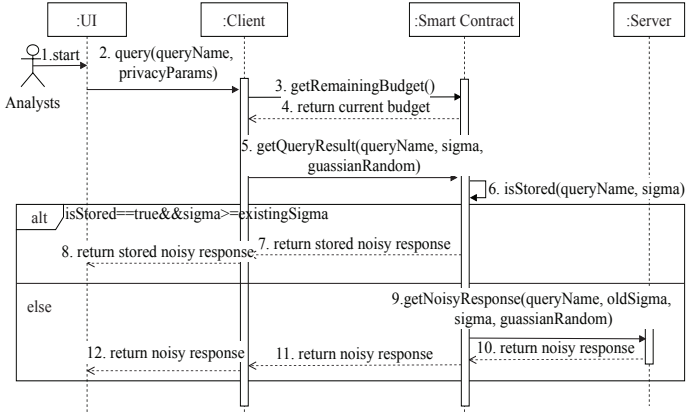}
\caption{The proposed blockchain-based system working flow  for differential-privacy  cost management.}
\label{fig:sequence}
\end{figure}

\section{Implementation and Experiments} \label{sec-Experiments}

In this section, we perform experiments to validate that the proposed system and algorithm are effective in saving privacy cost according to the system flow shown in Fig.~\ref{fig:sequence}. More specifically, a user sends a query through the UI, and then \texttt{Client} receives the query and forwards it to \texttt{Blockchain} smart contract. After the smart contract checks with stored data, it will decide whether to return the noisy response to \texttt{Client} directly or forward the request to \texttt{Server}. If \texttt{Server} receives the request, it will generate and return a noisy response  to the smart contract. 

\subsection{Experiment Setup}

We prototype a web application based on the system description in Section~\ref{sec-system}.  We use the Javascript language to write \texttt{Client}, whereas the Solidity language is for \texttt{Blockchain} smart contract. Besides, Web3 is used as the Javascript API to exchange information between  \texttt{Client} and \texttt{Blockchain} smart contract, and then Node.js and Express web framework are leveraged to set up \texttt{Server}. In addition, MongoDB is used as the database to host the real-world dataset. Our designed smart contracts are deployed on the Ropsten~\cite{ropsten} testnet with the MetaMask extension of the Chrome browser. The Ropsten testnet is a testing blockchain environment maintained by Ethereum, and it implements the same Proof-of-Work protocol as the main Ethereum network.  Fig.~\ref{fig:screen} shows the screenshot of our blockchain-based privacy management system. Fig.~\ref{fig:output} presents outputs while sending queries using the system. 


We evaluate the performance of the proposed differential privacy mechanism based on  a real-world dataset containing American community survey samples extracted from the \textit{Integrated Public Use Microdata Series} at \url{https://www.ipums.org}. There are 5000 records in the  dataset. Each record includes the following numerical attributes: ``Total personal income'', ``Total family income'', ``Age'', and categorical attributes: ``Race'', ``Citizenship status''. We set the privacy budget as $\epsilon_{\textup{budget}} = 8$ and $\delta_{\textup{budget}} = 10^{-4}$, which are commonly used to protect the privacy of a dataset~\cite{sanchez2014improving, wang2019collecting}. We consider five types of queries: ``average personal income'', ``average total family income'', ``frequency of US citizens'', ``frequency of white race'', and ``frequency of age more than 60''. For the privacy parameter of each query $Q_i$, we sample $\epsilon_i$ uniformly from $[0.1, 1.1]$ and sample $\delta_i$ uniformly from $[10^{-5}, 10^{-4}]$. The sensitivities of these queries are 202, 404, 0.0002, 0.0002, and 0.0002, 
respectively. We compute the sensitivity of a query based on  Section~\ref{subsec-ell2-sensitivity}.  For the query ``average total personal income'', since the user's total personal income ranges from $-5000$ to $700000$ in the dataset mentioned above, we assume the domain of total personal income is in the range of $[-10000, 1000000]$ for all possible datasets. The sensitivity is $(1000000-(-10000))/5000 = 202$ and the mechanism protects the privacy of all data within $[-10000, 1000000]$. Thus, it can protect the privacy of the dataset in our experiment. Suppose the received query is ``average total family income''. In that case, we assume the maximal variation is $[-20000, 2000000]$ for all possible datasets because the total family income's range is $[-5000,1379500]$ in the dataset we use. The sensitivity is $(2000000-(-20000))/5000=404$. Hence, our generated noise with the sensitivity of $404$ can protect the privacy of all data within $[-20000, 2000000]$. Therefore, it can protect the privacy of the dataset we use as well.  The sensitivity for queries ``frequency of US citizens'', ``frequency of white race'', and ``frequency of age more than 60'' is $1/5000 = 0.0002$.



\begin{figure}[!t]
\centering
\includegraphics[scale=0.3]{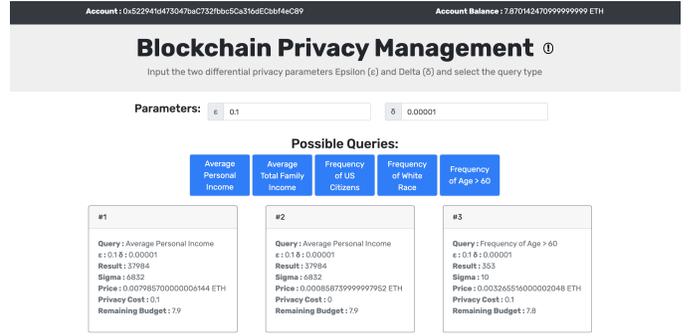}
\caption{Screenshot of blockchain-based privacy management system demo.}
\label{fig:screen}
\end{figure}

\begin{figure}[!t]
\centering
\includegraphics[scale=0.3]{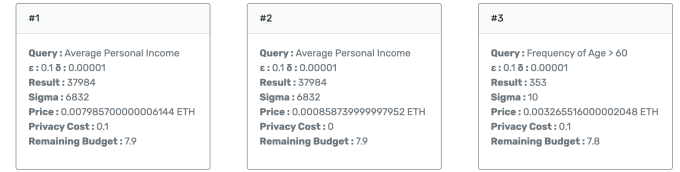}
\caption{Displaying of outputs with $\epsilon$ privacy cost.}
\label{fig:output}
\end{figure}




\begin{figure}[!t]
\centering
\includegraphics[scale=0.5]{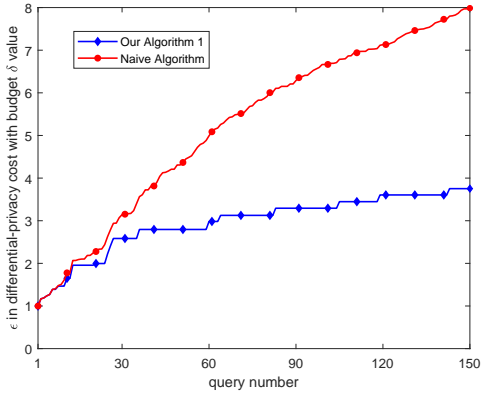}\vspace{-5pt}
\caption{Performance comparison of the sum of privacy cost.\vspace{-2pt}}
\label{fig:privacy_cost}
\end{figure}

\begin{figure}[!t]
\centering
\includegraphics[scale=0.5]{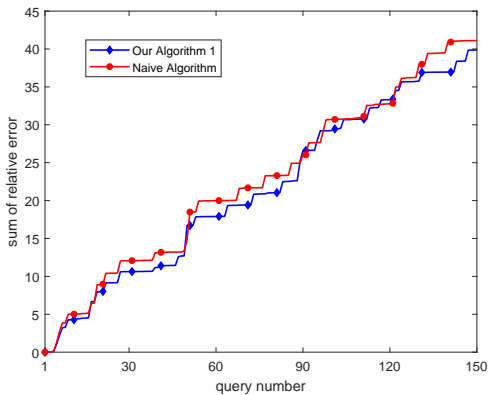}\vspace{-5pt}
\caption{Performance comparison of the sum of relative error.\vspace{-2pt}}
\label{fig:accuracy}
\end{figure}

\subsection{Experimental Results}

The benchmark of our experiment is a naive scheme which does not contain Algorithm~\ref{algmain} in the smart contract. That is, every query will be forwarded by the smart contract to \texttt{Server} to get the noisy response. Hence, no differential privacy cost can be reused in the naive scheme.

First, we use an experiment to validate that our proposed Algorithm~\ref{algmain} is effective in saving privacy cost. Thus, we design a performance comparison experiment by tracking privacy cost using our Algorithm~\ref{algmain} and the naive scheme, respectively. Specifically, we deploy two smart contracts implementing our Algorithm~\ref{algmain} and the naive scheme respectively on the Ropsten testnet. Then, we send 150 requests randomly selected in five query types from \texttt{Client} of the web application, and record the privacy cost of each query. As shown in Fig.~\ref{fig:privacy_cost},  compared with the naive scheme, the proposed algorithm saves significant privacy cost. 
When the number of the queries is 150, the differential-privacy cost of Algorithm~\ref{algmain} is about 52$\%$ less than that of the naive algorithm.  
We also observe that the privacy cost in the proposed scheme increases slowly when the number of queries increases, even trending to converge to a specific value. The reason is that, in Algorithm~\ref{algmain}, for each query type, we can always partially or fully reuse previous noisy answers when the query type is asked for a second time or more. Therefore, in our scheme, many queries are answered without incurring additional privacy cost if noisy responses fully reuse previous noisy answers.

Second, to prove that the proposed Algorithm~\ref{algmain} retains the accuracy of the dataset, we design another experiment to compare the sum of relative errors. We use the same smart contracts as those in the last experiment. We accumulate relative errors incurred in each query. Fig.~\ref{fig:accuracy} shows that the sum of relative errors of Algorithm~\ref{algmain} is comparable with that of the naive scheme. Since relative errors are similar between two schemes,  our results demonstrate that the proposed Algorithm~\ref{algmain} keeps the accuracy. 

As a summary, Fig.~\ref{fig:privacy_cost} and Fig.~\ref{fig:accuracy} together demonstrate that our Algorithm~\ref{algmain} can save privacy cost significantly without sacrificing the accuracy of the dataset. 

Third, we evaluate the latency of our system. The latency of our system is affected by the blockchain, MetaMask, network condition and the server. To shorten the speed of network transmission, we setup a local testnet at http://localhost:3000/ using Ganache-cli client with blockTime set as 15s~\cite{ganache2020}. Fig.~\ref{fig:latency} shows that the latency increases as the number of queries increases. The capacity of Ethereum's throughput is 20$\sim$60 Transmission Per Second (TPS)~\cite{gervais2016security,cao2019internet,vujivcic2018blockchain,stark2018making,tang2019public,miraz2019lapps}. When the number of queries reaches $60$, the latency increases significantly. In addition to Ethereum's throughput, both MetaMask and the capacity of the deployed device affect the latency. We test the case when the query results have been saved into the system. Thus, smart contract does not need to send requests to the server. We can obtain query results by using previous query results. The worst case is that smart contract has to send request every query which takes longer time to obtain the result because of the third party service Provable.

\begin{figure}[!h]
\centering
\includegraphics[scale=0.4]{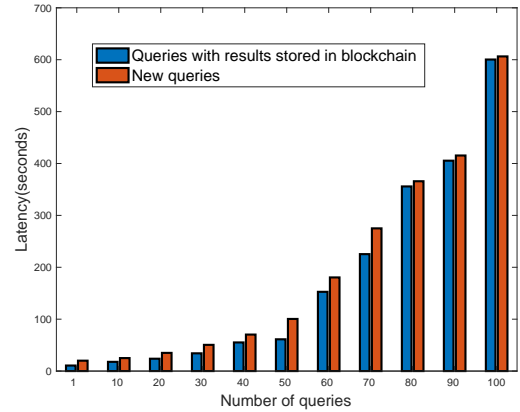}\vspace{-5pt}
\caption{Latency vs the number of queries.}
\label{fig:latency}
\end{figure}


Forth, we evaluate the relationship between the query utility and the privacy budget. As defined in~\cite{blum2013learning}, the privacy utility of a mechanism satisfies $(\alpha, \beta)$-useful if $|\widetilde{Q}_m(D)- Q_m(D)| \leq \alpha$ with probability at least $1 - \beta$. Thus, a small $\alpha$ means that the difference between the perturbed result and the actual result is small, which also reflects that the mechanism has a high utility. The noise added to a query can be calculated as ${\sigma} = \Gaussian(\Delta_{Q}, \epsilon, \delta)$, where $\Gaussian(\Delta_{Q}, \epsilon, \delta) \vspace{-0pt}=\sqrt{2\ln\frac{ 1.25}{\delta}}\times\frac{\Delta_{Q}}{\epsilon}$. We set $\delta = 10^{-5}$ and $\epsilon \in [1, 8]$. Appendix~\ref{sec-proof-theorem-4} proves that when we set $\beta = 0.05$, $\alpha = 2 \sigma$. Fig.~\ref{fig:utility}  and Fig.~\ref{fig:noise} illustrate how the utility and noise change as the privacy budget $\epsilon$ increases. Fig.~\ref{fig:utility} shows the value of $\alpha$ decreases when the privacy budget $\epsilon$ increases, meaning that the utility increases. In addition, the amount of noise added reflects the query utility as well. When less noise is added to the query response, the more utility the response gains. Fig.~\ref{fig:noise} shows that how the noise changes with the privacy budget. As the privacy budget increases, noise decreases, which means that the query utility increases. The amount of noise depends on the privacy budget  and the sensitivity value.  Queries such as ``Frequency of US citizens'', ``Frequency of white race'' and ``Frequency of age more than 60'' have the same sensitivity value $0.0002$, so the noise added to their responses is the same when the privacy budgets they use are equal.

\begin{figure}[!h]
\centering
\includegraphics[scale=0.4]{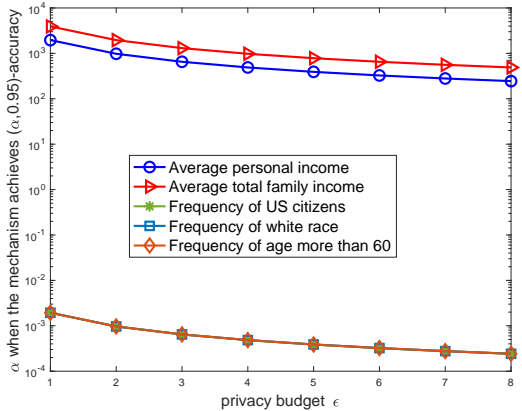}\vspace{-5pt}
\caption{Utility vs the privacy budget.\vspace{-10pt}}
\label{fig:utility}
\end{figure}

\begin{figure}[!h]
\centering
\includegraphics[scale=0.4]{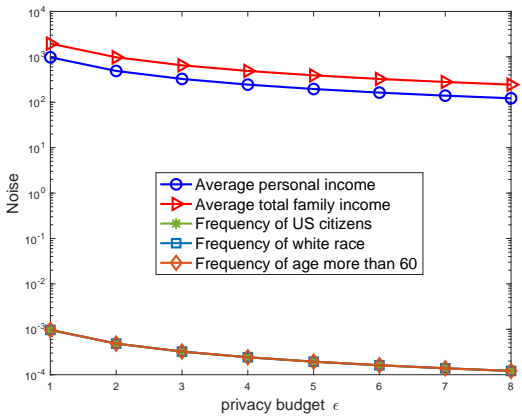}\vspace{-5pt}
\caption{Noise vs the privacy budget.\vspace{-10pt}}
\label{fig:noise}
\end{figure}

Fifth, we evaluate the efficiency of computing resources used by the DP scheme using or without the blockchain. The scheme without the blockchain's involvement is when data analysts call the server directly to obtain noisy answers.  We simulate this case using ``Apache JMeter'' API testing software to send queries with differential privacy parameters to the server directly. We consider the CPU usage as the metric for evaluating the efficiency of computing resources~\cite{zheng2018detailed, morishima2018accelerating}. 

\textbf{Experimental Settings:} We use a MacBook Pro with 2.3 GHz Quad-Core Intel Core i5, and 8 GB 2133 MHz LPDDR3  to run  applications. In the experiment,  we firstly deploy a locally private blockchain testnet using ``Go-Ethereum'' platform with three nodes mining. We then hold a separate server locally for testing  CPU usage for processing API requests sent from ``Apache JMeter''. We randomly send 50 queries for testing. At the same time, we use the ``Activity Monitor'' software to track and obtain their CPU usage~\cite{king2012ppcoin}.

\textbf{Experimental Results:} Fig.~\ref{fig:cpu} compares CPU usages spent by schemes with and without the blockchain. We can observe that the computation efficiency (i.e., CPU usage) in our blockchain-based scheme is higher than that using the server to handle directly. However, it is still acceptable. Node.js and Express.js web application frameworks, which we use to build the web server application, are CPU-intensive. One approach to save computing resources when using blockchain is to stop mining when there is not any incoming task. Only start mining when it is necessary. Therefore, our proposed scheme is efficient and practical with acceptable computation cost.

\begin{figure}[!h]
\centering
\includegraphics[scale=0.35]{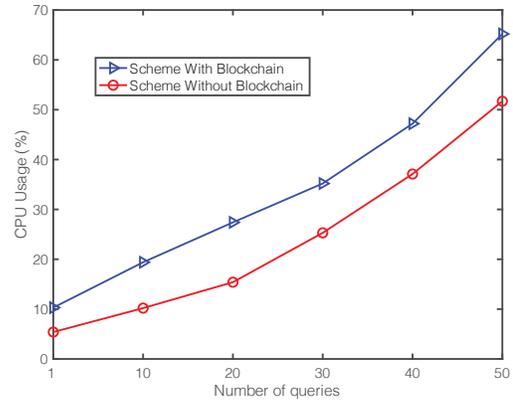}\vspace{-5pt}
\caption{Comparison of DP scheme's CPU usage with or without smart contract.\vspace{-10pt}}
\label{fig:cpu}
\end{figure}

\section{Related Work}\label{sec-related-work}

In this section, we first compare   our paper and a closely related study~\cite{yang2017differentially}, and then discuss other related work.

\subsection{Comparison with Yang~\emph{et~al.}~\cite{yang2017differentially}} \label{section-yang-drawbacks}
Yang~\emph{et~al.}~\cite{yang2017differentially} utilize blockchain and differential privacy technologies to achieve  the security and privacy protection during data sharing. 
Compared with~\cite{yang2017differentially}, we summarize the differences between our work and~\cite{yang2017differentially} as follows.
\begin{itemize}
\item Although Algorithm~\ref{algmain} of~\cite{yang2017differentially} claims to satisfy $\epsilon$-differential privacy, it does not since the noisy output's domain (i.e., the set of all possible values) depends on the input. The explanation is as follows. In~\cite{yang2017differentially}, for two neighboring datasets $D$ and $D'$,  there exists a subset $\mathcal{Y}$ of outputs such that $\bp{\widetilde{Q}(D)\in \mathcal{Y}} > 0$ but $\bp{\widetilde{Q}(D')\in \mathcal{Y}} = 0$. This \vspace{1pt} means $ \frac{\bp{\widetilde{Q}(D)\in \mathcal{Y}}}{\bp{\widetilde{Q}(D')\in \mathcal{Y}}} = \infty > e^{\epsilon}$, which violates $\epsilon$-differential privacy for any $\epsilon<\infty$.
\item \cite{yang2017differentially} does not discuss how to choose the small additional privacy parameter in its Algorithm~\ref{algmain}.
\item In~\cite{yang2017differentially}, when a query is asked for the first time, the Laplace mechanism of~\cite{dwork2006calibrating} for $\epsilon$-differential privacy is used to add Laplace noise to the true query result.  Afterwards,~\cite{yang2017differentially} adds new Laplacian noise on previous noisy output, which makes the new noisy response no longer follow Laplace distribution since the sum of independent Laplace random variables does not follow a Laplace distribution. Hence, the analysis in~\cite{yang2017differentially} is not effective.
\end{itemize}

We consider  $(\epsilon, \delta)$-differential privacy by using the Gaussian noise. The advantage of Gaussian noise over Laplace noise lies in the easier privacy analysis for the composition of different privacy-preserving algorithms, since   the sum of   independent Gaussian   random variables still follows the Gaussian distribution, while the sum of independent Laplace random variables does not obey a Laplace distribution.  




\subsection{Other Related Work}
Differential privacy, a strong mathematical model to guarantee the database's privacy, has attracted much attention in recent years. Blockchain is a fast-growing technology to provide security and privacy in a decentralized manner \cite{li2020mitigating, 8832210,8624307,8489897,tsai2019improved, fernandez2019atomic}. Feng~\emph{et~al.}~\cite{feng2019survey} summarize prior studies about privacy protection in blockchain system, including methodology for identity and transaction privacy preservation. 
In the following, we will introduce more recent studies utilizing blockchain or privacy techniques to provide privacy or security protection in identity, data, and transactions.

\textbf{Leveraging Blockchains for Identity Privacy/Security Protection.}
A few studies have focused on leveraging the blockchain to guarantee privacy/security in access control management or identity protection. For example, Zyskind~\emph{et~al.}~\cite{zyskind2015decentralizing} and Xia~\emph{et~al.}~\cite{xia2017bbds} both use blockchain in access control management. Zyskind~\emph{et~al.}~\cite{zyskind2015decentralizing} create a decentralized personal data management system to address users' concerns about privacy when using third-party mobile platforms.  Xia~\emph{et~al.}~\cite{xia2017bbds} propose a permissioned blockchain-based data sharing framework to allow only verified users to access the cloud data.
Lu~\emph{et~al.}~\cite{lu2018zebralancer} develop a private and anonymous decentralized crowdsourcing system ZebraLancer, which overcomes data leakage and identity breach in traditional decentralized crowdsourcing. The above studies focuse on identity privacy because the Blockchain is anonymous, whereas they do not consider the privacy protection for the database. 


\textbf{Leveraging Blockchains for Data Privacy/Security Protection.}
In addition to the identity privacy preservation, Hu~\emph{et~al.}~\cite{hu2018searching} replace the central server with a smart contract and construct a decentralized privacy-preserving search scheme for computing encrypted data while ensuring the privacy of data to prevent from misbehavings of a malicious centralized server. Luongo~\emph{et~al.}~\cite{luongokeep} use  secure multi-party computation to design a privacy primitive named Keep which allows contracts to manage and use private data without exposing the data to the public blockchain for protecting smart contracts on public blockchains. Alternatively, we use differential privacy standard to guarantee privacy. Moreover, blockchains are popular to be  used for  security protection of data sharing in  IoT scenarios~\cite{8624307,8489897,xia2017bbds}.

\textbf{Leveraging Blockchains for Transaction Privacy/Security Protection.}
Moreover, some previous studies use blockchain to guarantee security and privacy in transactions. For example, Henry~\emph{et~al.}~\cite{henry2018blockchain} propose that the blockchain should use mechanisms that piggyback on the overlay network, which is ready for announcing transactions to de-link users' network-level information instead of using an external service such as Tor to protect users' privacy. Gervais~\cite{gervais2016security} propose a quantitative framework to analyze the security of Proof-of-Work in blockchains, where the framework's inputs included security, consensus, and network parameters. Herrera-Joancomart{\'\i} and P{\'e}rez-Sol{\`a}~\cite{herrera2016privacy} focuse on privacy in bitcoin transactions. Sani~\emph{et~al.}~\cite{sani2019xyreum} propose a new blockchain Xyreum with high-performance and scalability to secure transactions in the Industrial Internet of Things. 


\textbf{Leveraging differential privacy for protecting privacy of IoT data.} IoT devices collect users' usage status periodically, which may contain sensitive information such as the energy consumption or location information. To avoid the privacy leakage, some studies use differential privacy mechanisms to protect the privacy of data~\cite{tudor2020bes,hassan2019privacy,xiong2018enhancing,liu2018epic,gai2019differential}. For example, Tudor~\emph{et~al.}~\cite{tudor2020bes} propose a streaming-based framework, Bes, to disclose IoT data. Hassan~\emph{et~al.}~\cite{hassan2019privacy} treat each IoT node as a node of blockchain to guarantee the IoT nodes' security, and they leverage differential privacy mechanisms to protect the privacy of data of each node.  To prevent adversaries  from intercepting the Internet traffic from/to the smart home gateway or profile residents’ behaviors through digital traces, Liu~\emph{et~al.}~\cite{liu2018epic} leverage differential privacy to develop a framework to prevent from attacks.  However, none of them discusses how to reuse the differential privacy budget.

\textbf{Differentially Private Algorithms.} Some differential privacy algorithms are proposed to provide privacy protection. Xiao~\emph{et~al.}~\cite{xiao2011ireduct} propose an algorithm to correlate the Lapalce noise added to different queries to improve the overall accuracy. Given a series of counting queries, the mechanism proposed by Li and Miklau~\cite{li2012adaptive}  select a  subset of queries to answer privately and use their noisy answers to derive answers for the remaining queries. For a set of non-overlapping counting queries, Kellaris and Papadopoulos~\emph{et~al.}~\cite{kellaris2013practical}  pre-process the counts by elaborate
grouping and smoothing them via averaging to reduce the sensitivity and thus the amount of injected noise. Given a workload of queries, Yaroslavtsev~\emph{et~al.}~\cite{yaroslavtsev2013accurate} introduce a solution to balance accuracy and efficiency by answering some queries more accurately than others.

\section{Discussion and Future Work}

In this section, we will discuss the meaning of differential privacy parameters, privacy of the  smart contract and queries.

\subsection{Differential Privacy Parameters}~\label{sec:dp-param-setting}

The value of differential privacy parameter $\epsilon$ represents the level of the protection provided by differential privacy mechanisms. McSherry~\emph{et al.}~\cite{mcsherry2010differentially, aaby2018privacy} quantify the strength of differential privacy as follows. When $\epsilon = 0$, the differential privacy mechanism provides perfect privacy. Then, when $\epsilon  \leq 0.1$, the protection is considered as strong, while $\epsilon  \geq 10$, the privacy protection is weak. The privacy parameter $\delta$ represents the small probability that a record gets altered in the database, so it should be very small. We sample $\delta$ uniformly from $[10^{-5},10^{-4}]$ for each query.

\subsection{Privacy for Smart Contract}

A smart contract is publicly available when it is deployed to the public blockchain. In our experiment, attackers can obtain the algorithm implemented in the smart contract. However, they still cannot obtain the accurate responses from noisy results even if they obtain the algorithm. There are some approaches that can be used to protect the privacy of the smart contract as follows :

First, Kosba~\emph{et al.}~\cite{kosba2016hawk} propose Hawk, a framework to build the privacy-preserving smart contracts. Hawk enables that programmers write private smart contracts without considering cryptography. The framework will generate a cryptographic protocol during compiling automatically.

Second, we can partially address this problem by deploying Ethereum to a private blockchain. We may combine the private blockchain and Proof-of-Authority~\cite{authority} consensus mechanism.  When Ethereum is deployed to the private blockchain, the private blockchain can set access control. Thus, the attackers need to break access control before accessing the smart contract. Therefore, when using a private blockchain, we consider access control to protect the smart contracts' privacy.

As it is complex to protect the smart contract's privacy, we would like to consider the smart contract's privacy as our future work.

\subsection{Privacy for Queries}

Differential privacy (DP) mechanisms consider that data analysts are untrusted, and the curator who holds the database is trusted. The trusted curator stores the database and responds to statistical queries made by an untrusted analyst so that DP will not protect the privacy of data analysts' queries. Moreover, DP supports statistical queries that may not include much sensitive information. When we use the smart contract, some data analysts may worry about the privacy of their queries. There are two ways to protect the privacy of queries :

First, we may use the private blockchain and conduct experiments using the private blockchain and Proof-of-Authority  consensus mechanism. Ethereum also supports deploying the smart contract to the private blockchain. In this case, the smart contract can be considered trusted, so that the sensitive information of queries will not leak.

Second, we may combine other cryptographic techniques with differential privacy. For example, Agarwal~\emph{et al.}~\cite{agarwal2019encrypted} design the encrypted databases (EDB) that support differentially-private statistical queries. In their paper, both the curator and the analyst are untrusted. The curator can outsource the database to an untrusted server securely by leveraging EDBs to encrypt operations.  Since the privacy protection for queries is complicated and may involve more privacy and security techniques, we would like to consider the privacy of queries as our future work.

Third, query privacy can be avoided by fixing the types of queries. Since the privacy budget is quite limited, it is impossible to let data analysts ask too many questions. Thus, one solution is to control types of queries. The system builder may take some time to select commonly used questions by data analysts, and then they set a dropdown list for data analysts to select questions. In this case, our system will not leak the queries' privacy because quries are standard.

\subsection{Difference between Our Proposed Scheme and Normal DP Schemes}

We pre-define queries for efficiently calculating the sensitivity values and saving users' time. The sensitivity values for different queries should be pre-defined even if we do not use blockchain.  When a query comes in many times with different differential privacy parameters, our scheme will play an essential role in saving the differential privacy budget. For example, many companies are trying to send the same query to a dataset because of the similar data analysis tasks. In this case, some of the privacy budgets can be saved. Since the privacy budget is a scarce resource regarding to the dataset, it is necessary to use our scheme.

However, when a query is seen for the first time, our scheme can only treat it as the new query. Since the differential privacy budget is quite limited, and sensitivities have to be calculated beforehand, a dataset will not support too many different statistical queries. If our system is implemented in the real world, similar to Fig.~\ref{fig:screen}, a list of queries supported maybe provided to control the variations of queries instead of letting data analysts type in different queries freely.

\section{Conclusion}\label{sec-conclusion}

In this paper, we use a blockchain-based approach for tracking and saving differential-privacy cost. In our design, we propose an algorithm that reuses noise fully or partially for different instances of the same query type to minimize the accumulated privacy cost. The efficiency of the algorithm is proved  via a rigorous mathematical proof. Moreover, we design a blockchain-based system for conducting real-world experiments to confirm the effectiveness of the proposed approach.



\bibliographystyle{IEEEtran}
\bibliography{references}

\appendix\label{appendix}
\subsection{\textbf{Proof of Theorem~\ref{thm-Alg1-explain-privacy-cost-with-r}}}~\label{sec-proof-theorem-1}
\noindent 
\begin{proof}

(i) As noted in the statement of Theorem~\ref{thm-Alg1-explain-privacy-cost-with-r}, we suppose that before answering query $Q_m$ and after answering $Q_1, Q_2, \ldots, Q_{m-1}$, the privacy loss $L_{\widetilde{Q}_1 \Vert  \widetilde{Q}_2 \Vert   \ldots  \Vert   \widetilde{Q}_{m-1}}(D, D')$ is given by $\mathcal{N}(\frac{A(D, D')}{2},A(D, D'))$ for some $A(D, D')$. Later we will show the existence of such $A(D, D')$. Then when $y_i$ follows the probability distribution of random variable $\widetilde{Q}_i(D)$ for each $i\in \{1,2,\ldots, m-1\}$, we have the following for the privacy loss $L_{\widetilde{Q}_1 \Vert  \widetilde{Q}_2 \Vert   \ldots  \Vert   \widetilde{Q}_{m-1}}(D, D';y_1, y_2, \ldots, y_{m-1})$:
\begin{align} 
&  L_{\widetilde{Q}_1 \Vert  \widetilde{Q}_2 \Vert   \ldots  \Vert   \widetilde{Q}_{m-1}}(D, D';y_1, y_2, \ldots, y_{m-1})   \nonumber
  \\ &:= 
  \ln \frac{ \mathbb{F}\left[\cap_{i=1}^{m-1} \left[\widetilde{Q}_i(D) = y_i\right]\right] }{  \mathbb{F}\left[\cap_{i=1}^{m-1} \left[\widetilde{Q}_i(D') = y_i \right]\right]} \sim \mathcal{N}(\frac{A(D, D')}{2},A(D, D')),\label{def-DP-loss-m1}
 \end{align}
where we use ``$\sim $'' to mean ``obeying the distribution''.

Now we need to analyze the privacy loss after answering the $m$ queries $Q_1, Q_2, \ldots, Q_{m}$. We look at  the privacy loss $L_{\widetilde{Q}_1 \Vert  \widetilde{Q}_2 \Vert   \ldots  \Vert   \widetilde{Q}_{m}}(D, D';y_1, y_2, \ldots, y_{m})$ defined as follows:
\begin{align} 
&  L_{\widetilde{Q}_1 \Vert  \widetilde{Q}_2 \Vert   \ldots  \Vert   \widetilde{Q}_{m}}(D, D';y_1, y_2, \ldots, y_{m})   \nonumber
  \\ &:= 
  \ln \frac{ \mathbb{F}\left[\cap_{i=1}^{m} \left[\widetilde{Q}_i(D) = y_i\right]\right] }{  \mathbb{F}\left[\cap_{i=1}^{m} \left[\widetilde{Q}_i(D') = y_i \right]\right]} .\label{def-DP-loss-m}
 \end{align}
Hence, we use~(\ref{def-DP-loss-m1}) to analyze~(\ref{def-DP-loss-m}). From~(\ref{reuse-expr2}), since $\widetilde{Q}_m(D)$ is generated by reusing $\widetilde{Q}_j(D)$ and generating additional noise (if necessary), where $j$ is an integer in $\{1,2,\ldots,m-1\}$ as noted in the statement of Theorem~\ref{thm-Alg1-explain-privacy-cost-with-r}, we have
\begin{align} 
&  L_{\widetilde{Q}_1 \Vert  \widetilde{Q}_2 \Vert   \ldots  \Vert   \widetilde{Q}_{m}}(D, D';y_1, y_2, \ldots, y_{m})   \nonumber
  \\ & = 
  \ln \frac{ \mathbb{F}\left[\cap_{i=1}^{m-1} \left[\widetilde{Q}_i(D) = y_i\right]\right] \fr{\widetilde{Q}_m(D) = y_m~|~\widetilde{Q}_j(D) = y_j} }{  \mathbb{F}\left[\cap_{i=1}^{m-1} \left[\widetilde{Q}_i(D') = y_i \right]\right] \fr{\widetilde{Q}_m(D') = y_m~|~\widetilde{Q}_j(D') = y_j}}  \nonumber
  \\ & = \ln \frac{ \mathbb{F}\left[\cap_{i=1}^{m-1} \left[\widetilde{Q}_i(D) = y_i\right]\right]  }{  \mathbb{F}\left[\cap_{i=1}^{m-1} \left[\widetilde{Q}_i(D') = y_i \right]\right] } \nonumber
  \\ & \quad + \ln \frac{  \fr{\widetilde{Q}_m(D) = y_m~|~\widetilde{Q}_j(D) = y_j} }{   \fr{\widetilde{Q}_m(D') = y_m~|~\widetilde{Q}_j(D') = y_j}}. \label{def-DP-loss-m4}
 \end{align}
 
 We now discuss the first term $ \ln \frac{ \mathbb{F}\left[\cap_{i=1}^{m-1} \left[\widetilde{Q}_i(D) = y_i\right]\right]  }{  \mathbb{F}\left[\cap_{i=1}^{m-1} \left[\widetilde{Q}_i(D') = y_i \right]\right] }$ and the second term $ \ln \frac{  \fr{\widetilde{Q}_m(D) = y_m~|~\widetilde{Q}_j(D) = y_j} }{   \fr{\widetilde{Q}_m(D') = y_m~|~\widetilde{Q}_j(D') = y_j}}$ in the last row of~(\ref{def-DP-loss-m4}). To begin with, from~(\ref{def-DP-loss-m1}), the first term in the last row of~(\ref{def-DP-loss-m4}) follows the Gaussian distribution $\mathcal{N}(\frac{A(D, D')}{2},A(D, D'))$. Next, we analyze the second term in the last row of~(\ref{def-DP-loss-m4}). 

When $\widetilde{Q}_j(D)$ and $\widetilde{Q}_m(D)$ take $y_j$ and $y_m$ respectively,  $\widetilde{Q}_j(D) - {Q}_j(D)$ and $\widetilde{Q}_m(D) - Q_m(D) - r [\widetilde{Q}_j(D) -Q_j(D) ]$ take the following defined $g_j$ and $g_m$ respectively:
\begin{align} 
g_j & := y_j - {Q}_j(D),  \\
g_m & := y_m - Q_m(D) - r [y_j -Q_j(D) ]. \label{newdefgm}
\end{align}  
For $D'$ being a neighboring dataset of $D$, we further define 
\begin{align} 
h_j & :=  {Q}_j(D) - {Q}_j(D'), \label{eqdefinehj-1} \\
h_m & := {Q}_m(D) - {Q}_m(D'), \label{eqdefinehm-1}
\end{align} 
so that
\begin{align} 
g_j + h_j & =  y_j - {Q}_j(D'), \\
g_m + h_m - r h_j & = y_m - Q_m(D') - r [y_j -Q_j(D') ].
\end{align} 
Note that $h_j$ and $h_m$ are the same since $Q_j$ and $Q_m$ are the same.
From the above analysis, we obtain :
\begin{align} 
 & \fr{\widetilde{Q}_m(D) = y_m~|~\widetilde{Q}_j(D) = y_j}  \nonumber
\\ & =  
 \fr{ \begin{array}{l}   \widetilde{Q}_m(D) - Q_m(D) - r [\widetilde{Q}_j(D) -Q_j(D) ]   = g_m \\ ~|~\widetilde{Q}_j(D) = y_j \end{array}}   \nonumber 
\\ & \stackrel{\textup{(b)}}{=}  \frac{1}{\sqrt{2\pi ({\sigma_m}^2 - r^2 {\sigma_j}^2)}} e^{-\frac{{g_m}^2}{2({\sigma_m}^2 - r^2 {\sigma_j}^2)}}  , 
\end{align} 
where step (b) follows since where $\widetilde{Q}_j(D) - {Q}_j(D)$ is a zero-mean Gaussian random variable with variance $\sigma_j^2$ and \mbox{$\widetilde{Q}_m(D) - Q_m(D) - r [\widetilde{Q}_j(D) -Q_j(D) ]$}  is a zero-mean Gaussian random variable with variance ${\sigma_m}^2 - r^2 {\sigma_j}^2$.

Similarly, for dataset $D'$, we have :
\begin{align} 
 &\fr{\widetilde{Q}_m(D') = y_m~|~\widetilde{Q}_j(D') = y_j}   \nonumber
\\ &  =  
\fr{ \begin{array}{l}   \widetilde{Q}_m(D') - Q_m(D') - r [\widetilde{Q}_j(D') -Q_j(D') ]  \\ = g_m + h_m - r h_j \\ ~|~\widetilde{Q}_j(D') = y_j   \end{array}} \nonumber 
\\ & \stackrel{\textup{(b)}}{=}  \frac{1}{\sqrt{2\pi ({\sigma_m}^2 - r^2 {\sigma_j}^2)}} e^{-\frac{(g_m + h_m - r h_j)^2}{2({\sigma_m}^2 - r^2 {\sigma_j}^2)}}, 
\end{align} 
where step (b) follows since where $\widetilde{Q}_j(D') - {Q}_j(D')$ is a Gaussian random variable with variance $\sigma_j^2$ and \mbox{$\widetilde{Q}_m(D') - Q_m(D') - r [\widetilde{Q}_j(D') -Q_j(D')$}  is a zero-mean Gaussian random variable with variance ${\sigma_m}^2 - r^2 {\sigma_j}^2$.

Then,
\begin{align} 
& \ln \frac{\fr{\widetilde{Q}_m(D) = y_m~|~\widetilde{Q}_j(D) = y_j}  }{\fr{\widetilde{Q}_m(D') = y_m~|~\widetilde{Q}_j(D') = y_j}  } \nonumber
\\ & = \ln \frac{~~  \frac{1}{\sqrt{2\pi ({\sigma_m}^2 - r^2 {\sigma_j}^2)}} e^{-\frac{{g_m}^2}{2({\sigma_m}^2 - r^2 {\sigma_j}^2)}}~~}{~~ \frac{1}{\sqrt{2\pi ({\sigma_m}^2 - r^2 {\sigma_j}^2)}} e^{-\frac{(g_m + h_m - r h_j)^2}{2({\sigma_m}^2 - r^2 {\sigma_j}^2)}} ~~}  \nonumber
\\ & =  \frac{(g_m + h_m - r h_j)^2 - {g_m}^2}{2({\sigma_m}^2 - r^2 {\sigma_j}^2)} \nonumber \\ & =    \frac{g_m (h_m - r h_j)}{{\sigma_m}^2 - r^2 {\sigma_j}^2}  + \frac{(h_m - r h_j)^2}{2({\sigma_m}^2 - r^2 {\sigma_j}^2)}.~\label{eq:proof-reuse}
\end{align}  

The above~(\ref{eq:proof-reuse}) presents the second term in the last row of~(\ref{def-DP-loss-m4}). At first glance, it may seem that the first term $ \ln \frac{ \mathbb{F}\left[\cap_{i=1}^{m-1} \left[\widetilde{Q}_i(D) = y_i\right]\right]  }{  \mathbb{F}\left[\cap_{i=1}^{m-1} \left[\widetilde{Q}_i(D') = y_i \right]\right] }$ and the second term $ \ln \frac{  \fr{\widetilde{Q}_m(D) = y_m~|~\widetilde{Q}_j(D) = y_j} }{   \fr{\widetilde{Q}_m(D') = y_m~|~\widetilde{Q}_j(D') = y_j}}$ in the last row of~(\ref{def-DP-loss-m4}) are dependent since they both involve $y_j$. However, we have shown from~(\ref{eq:proof-reuse}) above that the second term in the last row of~(\ref{def-DP-loss-m4}) depends on only the random variable $g_m$ (note that terms in~(\ref{eq:proof-reuse}) other than $g_m$ are all given), which is the amount of additional Gaussian noise used to generated $\widetilde{Q}_m(D)$ according to~(\ref{reuse-expr2}) and~(\ref{newdefgm}); i.e., the second term in the last row of~(\ref{def-DP-loss-m4}) is actually independent of the first term in the last row of~(\ref{def-DP-loss-m4}). From~(\ref{def-DP-loss-m1}), the first term in the last row of~(\ref{def-DP-loss-m4}) follows the Gaussian distribution $\mathcal{N}(\frac{A(D, D')}{2},A(D, D'))$. Next, we show that~(\ref{eq:proof-reuse}) presenting the second term in the last row of~(\ref{def-DP-loss-m4}) also follows a Gaussian distribution. 

Since $g_m$ follows a zero-mean Gaussian distribution with variance ${\sigma_m}^2 - r^2 {\sigma_j}^2$, clearly $ \frac{g_m (h_m - r h_j)}{{\sigma_m}^2 - r^2 {\sigma_j}^2}$ follows a zero-mean Gaussian distribution with variance given by
\begin{align}
&  \bigg[\frac{ (h_m - r h_j)}{{\sigma_m}^2 - r^2 {\sigma_j}^2}\bigg]^2 \times ({\sigma_m}^2 - r^2 {\sigma_j}^2)  \nonumber  \\ &  =  \frac{(h_m - r h_j)^2}{{\sigma_m}^2 - r^2 {\sigma_j}^2}. \label{defineV-1}
\end{align}

Since $Q_m$ and $Q_j$ are the same, we obtain from Eq.~(\ref{eqdefinehj-1}) and Eq.~(\ref{eqdefinehm-1}) that $h_j = h_m = {Q}_m(D) - {Q}_m(D')$, which we use to write Eq.~(\ref{defineV-1}) as
\begin{align}
\frac{[\| Q_m(D) -Q_m(D') \|_2]^2 (1 - r )^2}{{\sigma_m}^2 - r^2 {\sigma_j}^2}. \label{defineV-v2}
\end{align}
Summarizing the above, privacy loss is
\begin{align} 
B_r(D, D'): = A(D, D') + \frac{[\| Q_m(D) -Q_m(D') \|_2]^2 (1 - r )^2}{{\sigma_m}^2 - r^2 {\sigma_j}^2}. ~\label{eq:proof-reuseBA}
\end{align}

As noted in the statement of Theorem~\ref{thm-Alg1-explain-privacy-cost-with-r}, we suppose that before answering query $Q_m$ and after answering $Q_1, Q_2, \ldots, Q_{m-1}$, the privacy loss $L_{\widetilde{Q}_1 \Vert  \widetilde{Q}_2 \Vert   \ldots  \Vert   \widetilde{Q}_{m-1}}(D, D')$ is given by $\mathcal{N}(\frac{A(D, D')}{2},A(D, D'))$ for some $A(D, D')$. With the above result~(\ref{eq:proof-reuseBA}), we can actually show that there indeed exists such $A(D, D')$. This follows from mathematical induction. For the base case; i.e., when only one query is answered, the result follows from Lemma 3 of~\cite{balle2018improving}. The induction step is given by the above result~(\ref{eq:proof-reuseBA}). Hence, we have shown the   existence of   $A(D, D')$. With this result and~(\ref{eq:proof-reuseBA}), we have completed proving Result~(i) of  Theorem~\ref{thm-Alg1-explain-privacy-cost-with-r}.

(ii) The optimal $r$ is obtained by minimizing $B_r(D, D')$ and hence minimizing $\frac{ (1 - r )^2}{{\sigma_m}^2 - r^2 {\sigma_j}^2}$. Analyzing the monotonicity of this expression, we derive the optimal $r$ as in Eq.~(\ref{roptimal}).
The first-order derivative of $B_r(D, D')$ to $r$ is:
\begin{align}
   B_r(D, D')'= \frac{-2(r{\sigma_j}^2-{\sigma_m}^2)(r-1)}{(r^2{\sigma_j}^2-{\sigma_m}^2)^2}.~\label{eq:first-order-B}
\end{align}
\begin{itemize}
    \item Case 1: if $\sigma_m \geq \sigma_j$, $B_r(D, D')' \geq 0$ when $r \in [1, \frac{\sigma_m}{\sigma_j}]$, and $B_r(D, D')' < 0$ when $r \in (-\infty, 1) \cup (\frac{\sigma_m}{\sigma_j}, +\infty)$. Hence, the optimal $r$ to minimize $B_r(D, D')$ is at $r = 1$.
    \item Case 2: if $\sigma_m < \sigma_j$, $B_r(D, D')' \geq 0$ when $r \in [\frac{\sigma_m}{\sigma_j},1]$, and $B_r(D, D')' < 0$ when $r \in (-\infty, \frac{\sigma_m}{\sigma_j}) \cup (1, +\infty)$. Hence, the optimal $r$ to minimize $B_r(D, D')$ is at $r =(\frac{\sigma_m}{\sigma_j})^2$.
\end{itemize}
Thus, we obtain optimal values of $r$ as Eq.~(\ref{roptimal}).
\end{proof}

\subsection{\textbf{Proof of Lemma~\ref{lem-privacy-cost-DP-condition}} }~\label{sec-proof-lemma-2}
\noindent 
\begin{proof}
Consider a query $R$ with $\ell_2$-sensitivity being $1$. Let $\widetilde{R}$ be the mechanism of adding Gaussian noise amount $\mu : = \frac{1}{\sqrt{\max_{\textup{neighboring datasets $D,D'$}} V(D, D')}} $ to $R$. From Corollary~\ref{cor-Alg1-explain-privacy-cost-with1}, the privacy loss of randomized mechanism $\widetilde{R}$ with respect to neighboring datasets $D$ and $D'$ is given by $\mathcal{N}(\frac{U(D, D')}{2},U(D, D'))$ for  $U(D, D'): = \frac{[\| R(D) -R(D') \|_2]^2}{{\mu}^2} $. By considering the $\ell_2$-sensitivity of $R$ (i.e., $\| R(D) -R(D') \|_2$) as $1$, $\max_{\textup{neighboring datasets $D,D'$}} V(D, D')$ and $\max_{\textup{neighboring datasets $D,D'$}} U(D, D')$ are the same. In addition, from Theorem 5 of~\cite{balle2018improving}, letting $Y$ (resp., $\widetilde{R}$) satisfy $(\epsilon, \delta)$-differential privacy can be converted to a condition on $\max_{\textup{neighboring datasets $D,D'$}} V(D, D')$ (resp., $\max_{\textup{neighboring datasets $D,D'$}} U(D, D')$).  Then letting $Y$ satisfy $(\epsilon, \delta)$-differential privacy is the same as letting $\widetilde{R}$ satisfy $(\epsilon, \delta)$-differential privacy. From Lemma~\ref{lemma-Gaussian}, $\widetilde{R}$ achieves $(\epsilon, \delta)$-differential privacy with $\mu = \Gaussian(1, \epsilon, \delta)$; i.e., if $\max_{\textup{neighboring datasets $D,D'$}} V(D, D') =  [\Gaussian(1, \epsilon, \delta)]^{-2} $. Summarizing the above, we complete proving Lemma~\ref{lem-privacy-cost-DP-condition}.
\end{proof}

\subsection{\textbf{Proof of Theorem~\ref{Alg1-explain-privacy-cost}} }~\label{sec-proof-theorem-2}
\noindent 
\begin{proof}
We use Theorem~\ref{thm-Alg1-explain-privacy-cost-with-r} to show Results~\ding{172}~\ding{173} and~\ding{174} of Theorem~\ref{Alg1-explain-privacy-cost}. Proof of~\ding{172}: In Case 2A) and Case 2C), $Q_m$ can reuse previous noise. Hence, the privacy loss will still be $\mathcal{N}(\frac{A(D, D')}{2},A(D, D'))$ according to Eq.~(\ref{Brresult}).\\
Proof of~\ding{173}: In Case 1), $Q_m$ cannot reuse previous noisy answers, and the new noise follows $\mathcal{N}(0, \sigma_m)$. Thus, $B(D, D'):=A(D, D') + \frac{[\| Q_m(D) -Q_m(D') \|_2]^2}{{\sigma_m}^2}$.\\
Proof of~\ding{174}: In Case 2B), $Q_m$ can reuse previous noisy answers partially, so we can prove it using Eq.~(\ref{Brresult}). 

Then, Lemma~\ref{lem-privacy-cost-DP-condition} further implies Results~\ding{175}~\ding{176} and~\ding{177} of Theorem~\ref{Alg1-explain-privacy-cost}.\\
Proof of~\ding{175}:  $Q_m$ can fully reuse the old noisy result in Cases~2A) and 2C). Thus, the privacy level does not change. \\
Proof of~\ding{176}: From Lemma~\ref{lem-privacy-cost-DP-condition}, we have $\hspace{-2pt}\max_{\textup{neighboring datasets $D,D'$}} A(D, D')\hspace{-2pt} =\hspace{-2pt}  [\Gaussian(1, \epsilon_{\textup{old}}, \delta_{\textup{budget}})]^{-2}$ and
\mbox{$\max_{\textup{neighboring datasets $D,D'$}} \hspace{-2pt}\bigg\{\hspace{-2pt} A(D, D')\hspace{-1pt} +\hspace{-1pt} [\| Q_m(D) \hspace{-1pt}-\hspace{-1pt}Q_m(D') \|_2]^2 $}$ \\ \times \hspace{-2pt} \frac{1}{{\sigma_m}^2} \bigg\}\hspace{-2pt} = \hspace{-2pt} [\Gaussian(1, \epsilon_{\textup{new}}, \delta_{\textup{budget}})]^{-2}$. The above two equations yield
\mbox{$  [\Gaussian(1, \epsilon_{\textup{new}}, \delta_{\textup{budget}})]^{-2} - [\Gaussian(1, \epsilon_{\textup{old}}, \delta_{\textup{budget}})]^{-2} $} \\$ =\max_{\textup{neighboring datasets $D,D'$}} [\| Q_m(D) -Q_m(D') \|_2]^2 \times \frac{1}{{\sigma_m}^2} = {\Delta_{Q_m}}^2 \times  \frac{1}{{\sigma_m}^2} = {\sigma_m}^2$. Hence, Gaussian($\Delta_{Q_m}$,$\epsilon\textup{\_squared\_cost}$,$\delta_{\textup{bugdet}}$) = $\sigma_m$.\\
Proof of~\ding{177}: From Lemma~\ref{lem-privacy-cost-DP-condition}, we have $\max_{\textup{neighboring datasets $D,D'$}} A(D, D') =  [\Gaussian(1, \epsilon_{\textup{old}}, \delta_{\textup{budget}})]^{-2} $ and
\mbox{$\max_{\textup{neighboring datasets $D,D'$}} \hspace{-2pt}\bigg\{\hspace{-2pt} A(D, D')\hspace{-1pt} +\hspace{-1pt} [\| Q_m(D) \hspace{-1pt}-\hspace{-1pt}Q_m(D') \|_2]^2 $}  $ \times \left[ \frac{1}{{\sigma_m}^2} - \frac{1}{[\min(\boldsymbol{\Sigma}_t)]^2}\right] \bigg\}=  [\Gaussian(1, \epsilon_{\textup{new}}, \delta_{\textup{budget}})]^{-2} $.
The above two equations yield
\mbox{$  [\Gaussian(1, \epsilon_{\textup{new}}, \delta_{\textup{budget}})]^{-2} - [\Gaussian(1, \epsilon_{\textup{old}}, \delta_{\textup{budget}})]^{-2} $} \\$ =\max_{\textup{neighboring datasets $D,D'$}} [\| Q_m(D) -Q_m(D') \|_2]^2 \times \left[ \frac{1}{{\sigma_m}^2} - \frac{1}{[\min(\boldsymbol{\Sigma}_t)]^2}\right] = {\Delta_{Q_m}}^2 \times \left[ \frac{1}{{\sigma_m}^2} - \frac{1}{[\min(\boldsymbol{\Sigma}_t)]^2}\right]  $. Then using the expression of $\Gaussian(\Delta_{Q}, \epsilon, \delta)$ from Lemma~\ref{lemma-Gaussian}, we further obtain Result~\ding{177}.
\end{proof}

\subsection{\textbf{Proof of Theorem~\ref{Alg1-total-privacy-cost}} }~\label{sec-proof-theorem-3}

\noindent 
\begin{proof}

First, from Theorem~\ref{Alg1-explain-privacy-cost}, after Algorithm~\ref{algmain} is used to answer all $n$ queries with query $Q_i$ being answered under $(\epsilon_i, \delta_i)$-differential privacy, the total privacy loss with respect to neighboring datasets $D$ and $D'$ is given by $\mathcal{N}(\frac{G(D, D')}{2},G(D, D'))$ for some $G(D, D')$.

Next, we use Theorem~\ref{Alg1-explain-privacy-cost} to further show that the expression of $G(D, D')$ is given by Eq.~(\ref{label-eq-G}). From Theorem~\ref{Alg1-explain-privacy-cost}, among all queries, only queries belonging to Cases~1) and~2B) contribute to $G(D, D')$. Below we discuss the contributions respectively.

With $N_1$ denoting the set of $i \in \{1,2,\ldots,n\}$ such that $Q_i$ is in Cases~1), we know from Result~\ding{173} of Theorem~\ref{Alg1-explain-privacy-cost} that the contributions of queries in Cases~1) to $G(D, D')$ is  given by
\begin{align}
    \sum_{i \in N_1}\frac{[\| Q_i(D) -Q_i(D') \|_2]^2}{{\sigma_i}^2}.~\label{eq-n-1-loss}
\end{align}


Below we use Result~\ding{174} of Theorem~\ref{Alg1-explain-privacy-cost} to compute the contributions of queries in Case~2B) to $G(D, D')$. For $T_{2B}$ being the set of query types in Case~2B), we discuss each query type $t\in T_{2B}$ respectively.

From Result~\ding{174} of Theorem~\ref{Alg1-explain-privacy-cost}, the contribution  to $G(D, D')$ by answering $Q_{j_{t,1}}$ under differential privacy is
\begin{align}
     [\| Q_{j_{t,1}}(D) -Q_{j_{t,1}}(D') \|_2]^2 \left(\hspace{-1pt}  \frac{1}{\sigma_{j_{t,1}}^2} -  \frac{1}{\sigma_{j_{t,0}}^2} \right). \nonumber
\end{align}
Similarly, the contribution  to $G(D, D')$ by answering $Q_{j_{t,2}}$ under differential privacy is
\begin{align}
     [\| Q_{j_{t,2}}(D) -Q_{j_{t,2}}(D') \|_2]^2 \left(\hspace{-1pt}  \frac{1}{\sigma_{j_{t,2}}^2} -  \frac{1}{\sigma_{j_{t,1}}^2} \right). \nonumber
\end{align}
Similar analyses are repeated for additional type-$t$ queries in Case~2B). In particular, for each $s \in \{1,2,\ldots, m_t\}$, the contribution  to $G(D, D')$ by answering $Q_{j_{t,s}}$ under differential privacy is  
\begin{align}
     [\| Q_{j_{t,s}}(D) -Q_{j_{t,s}}(D') \|_2]^2 \left(\hspace{-1pt}  \frac{1}{\sigma_{j_{t,s}}^2} -  \frac{1}{\sigma_{j_{t,s-1}}^2} \right).\label{eq-contribution-s}
\end{align}
Summing all~(\ref{eq-contribution-s}) for $s \in \{1,2,\ldots, m_t\}$, we obtain that for each query type $t\in T_{2B}$, the contributions  to $G(D, D')$ by answering $Q_{j_{t,1}}, Q_{j_{t,2}}, \ldots, Q_{j_{t,m_t}}$ under differential privacy is  
\begin{align}
\sum_{s \in \{1,2,\ldots, m_t\}} [\| Q_{j_{t,s}}(D) -Q_{j_{t,s}}(D') \|_2]^2 \left(\hspace{-1pt}  \frac{1}{\sigma_{j_{t,s}}^2} -  \frac{1}{\sigma_{j_{t,s-1}}^2} \right)  .~\label{eq-sum-2b-1}
\end{align}

Since  $Q_{j_{t,0}}, Q_{j_{t,1}}, \ldots, Q_{j_{t,m_t}}$ for $j_{t,0},j_{t,1},\ldots ,j_{t,m_t}$ are all type-$t$ queries, $\| Q_{j_{t,s}}(D) -Q_{j_{t,s}}(D') \|_2$ are all the same for $s \in \{1,2,\ldots, m_t\}$. Hence, we write~(\ref{eq-sum-2b-1}) as
\begin{align}
   & \sum_{s \in \{1,2,\ldots, m_t\}}  \bigg\{  \frac{[\| Q_{j_{t,s}}(D) -Q_{j_{t,s}}(D') \|_2]^2}{\sigma_{j_{t,s}}^2} \nonumber \\&\quad \quad\quad \quad\quad \quad-  \frac{[\| Q_{j_{t,s-1}}(D) -Q_{j_{t,s-1}}(D') \|_2]^2}{\sigma_{j_{t,s-1}}^2} \bigg\} \nonumber \\
   & \quad=  \frac{{[\| Q_{j_{t,m_t}}(D) -Q_{j_{t,m_t}}(D') \|_2]}^2}{{\sigma_{j_{t,m_t}}^2}} \nonumber \\&\quad \quad \quad \quad \quad - \frac{{[\| Q_{j_{t,0}}(D) -Q_{j_{t,0}}(D') \|_2]}^2}{{\sigma_{j_{t,0}}^2}} . ~\label{eq-sum-2b}
\end{align}
Summing all~(\ref{eq-sum-2b}) for $t\in T_{2B}$, the contributions  to $G(D, D')$ by answering all queries in Case~2B) is
\begin{align}
   &   \sum_{t\in T_{2B}} \Bigg\{\frac{{[\| Q_{j_{t,m_t}}(D) -Q_{j_{t,m_t}}(D') \|_2]}^2}{{\sigma_{j_{t,m_t}}^2}} \nonumber
\\ & \quad \quad \quad    - \frac{{[\| Q_{j_{t,0}}(D) -Q_{j_{t,0}}(D') \|_2]}^2}{{\sigma_{j_{t,0}}^2}}\Bigg\} .  \label{eq-sum-2b2}
\end{align}
Then $G(D, D')$ as the sum of (\ref{eq-n-1-loss}) and (\ref{eq-sum-2b2}) is given by Eq.~(\ref{label-eq-G}). 

Summarizing the above, we have proved that after Algorithm~\ref{algmain} is used to answer all $n$ queries under differential privacy, the total privacy loss with respect to neighboring datasets $D$ and $D'$ is given by $\mathcal{N}(\frac{G(D, D')}{2},G(D, D'))$ for $G(D, D')$ in Eq.~(\ref{label-eq-G}). Furthermore, under
\begin{align}
  \max_{\textup{neighboring datasets $D,D'$}} \| Q_i(D) -Q_i(D') \|_2 =
 \Delta_{Q_i}  \nonumber
\end{align}
and
\begin{align}
    &\max_{\textup{neighboring datasets $D,D'$}} \| Q_{j_{t,m_t}}(D) -Q_{j_{t,m_t}}(D') \|_2 \nonumber \\&= 
   \max_{\textup{neighboring datasets $D,D'$}} \| Q_{j_{t,0}}(D) -Q_{j_{t,0}}(D') \|_2  = \Delta(\textup{type-}t),\nonumber
\end{align}
we use Eq.~(\ref{label-eq-G}) to have $\max_{\textup{neighboring datasets $D,D'$}} G(D, D') $ given by Eq.~(\ref{label-eq-G-maxDDprime}).

Finally, from Lemma~\ref{lem-privacy-cost-DP-condition}, the total privacy cost of our Algorithm~\ref{algmain} can be given by $(\epsilon_{\textup{ours}}, \delta_{\textup{budget}})$-differential privacy for $\epsilon_{\textup{ours}}$  satisfying
\begin{align} 
\hspace{-12pt}[\Gaussian(1, \epsilon_{\textup{ours}}, \delta_{\textup{budget}})]^{-2}= \max_{\textup{neighboring datasets $D,D'$}} G(D, D') ,\nonumber  
\end{align}  
or  $(\epsilon, \delta)$-differential privacy for any $\epsilon$ and $\delta$ satisfying  $[\Gaussian(1, \epsilon, \delta)]^{-2}= \max_{\textup{neighboring datasets $D,D'$}} G(D, D') $. 
\end{proof}

\subsection{\textbf{Utility of the Gaussian Mechanism.}}~\label{sec-proof-theorem-4}

\noindent 
\begin{proof} The noisy response for one-dimensional query $Q_m$ is $\widetilde{Q}_m(D) = Q_m(D) + N(0, \sigma^2)$. Letting the probability of $\| \widetilde{Q}_m(D) - Q_m(D) \|_p \leq \alpha$ be $1-\beta$, then we have
        \begin{align}
            1-\beta &= \pr{\| \widetilde{Q}_m(D) - Q(D) \|_p \leq \alpha} \nonumber \\
                &= \pr{|N(0, \sigma^2)| \leq \alpha } \nonumber \\ 
                &= \pr{ -\alpha \leq N(0, \sigma^2) \leq \alpha } \nonumber \\ 
                &= \pr{N(0, \sigma^2) \leq \alpha } - \pr{ N(0, \sigma^2) \leq -\alpha} \nonumber \\ 
                & = \frac{1}{2} \left[ 1+ \text{erf}\left(\frac{\alpha}{\sigma \sqrt{2}}\right)\right]-\frac{1}{2} \left[ 1+ \text{erf}{\left(\frac{-\alpha}{\sigma \sqrt{2}}\right)}\right] \nonumber \\
                & = \text{erf}\left(\frac{\alpha}{\sigma \sqrt{2}}\right), \label{eq-erf}
        \end{align}
where erf$(\cdot)$ denotes the error function and the last step of Eq.~(\ref{eq-erf}) uses the fact that erf$(\cdot)$ is an odd function.

According to the two-sigma rule of Gaussian distribution~\cite{pukelsheim1994three}, which can also be obtained from above equation that  $95\%$ values lie within two standard deviations of the mean. Thus, if we set $\alpha = 2 \sigma$, $\beta \approx 0.05$.
\end{proof}

\end{document}